\documentclass{article}
\pdfoutput=1
\usepackage{fullpage}
\usepackage{amsmath}
\usepackage{amssymb}
\usepackage{amsthm}
\usepackage{enumerate}
\usepackage{enumitem}
\usepackage[round]{natbib}
\usepackage[usenames,dvipsnames,svgnames,table]{xcolor}
\usepackage{tikz}
\usepackage{tikz-cd}
\usepackage{array}
\usetikzlibrary{arrows,shapes,shapes.misc,positioning,calc}
\usepackage[justification=centering]{subcaption}
\usepackage[colorinlistoftodos]{todonotes}
\usepackage{url}
\usepackage{hyperref}

\newcommand{\doi}[1]{doi:\href{https://doi.org/#1}{\nolinkurl{#1}}}

\usepackage[capitalise]{cleveref}
\usepackage{bera}
\usepackage{soul}
\usetikzlibrary{decorations.markings}
\usepackage{booktabs}
\usepackage{makecell}
\usepackage{multirow}
\usepackage{pifont}
\newcommand{\xmark}{\ding{55}}

\newcommand*\circled[1]{\tikz[baseline=(char.base)]{\node[shape=circle,draw,inner sep=1.5pt] (char) {$#1$};}}
\newcommand{\B}{\mathbb{B}}
\newcommand\vd{1.4}

\newcommand\td{0.9}
\newcommand\sd{6}

\newcommand{\rd}[1]{\rho(#1)}

\newcommand\subs{\Sigma}

\newcommand{\dyn}{\mathrm{D}}
\newcommand{\D}{\mathrm{D}}
\newcommand{\AD}{\mathrm{AD}}
\newcommand{\GAD}{\mathrm{GD}}
\newcommand{\SD}{\mathrm{SD}}

\newtheorem{proposition}{Proposition}[section]
\newtheorem{theorem}[proposition]{Theorem}

\newtheorem{lemma}[proposition]{Lemma}

\theoremstyle{definition}
\newtheorem{definition}[proposition]{Definition}
\newtheorem{remark}[proposition]{Remark}
\newtheorem{example}[proposition]{Example}
\newtheorem{question}{Question}

\tikzset{notImplies/.style={-{Implies},decoration={markings, mark = at position 0.5 with {
         \node (tempnode) {$\slash$};}
         }, postaction={decorate}}}

\begin{document}
\title{Phenotype control and elimination\\of variables in Boolean networks}
\author{Elisa Tonello$^1$ 
and Loïc Paulev\'{e}$^2$}
\date{\begin{small}$^1$ Freie Universität Berlin, Germany\\
\texttt{elisa.tonello@fu-berlin.de}\\
$^2$ Univ. Bordeaux, CNRS, Bordeaux INP, LaBRI, UMR 5800, F-33400 Talence, France\\
\texttt{loic.pauleve@labri.fr}
\end{small}}

\maketitle

\begin{tikzpicture}[remember picture,overlay]
\hypersetup{hidelinks}
\node[anchor=north west,yshift=220pt,xshift=-40pt]
{\href{https://doi.org/10.24072/pci.mcb.100318}{\includegraphics[height=35mm]{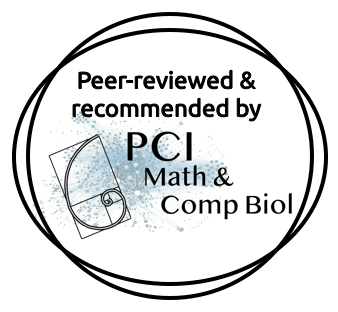}}};
\end{tikzpicture}

\begin{abstract}
  We investigate how elimination of variables can affect the asymptotic dynamics and
  phenotype control of Boolean networks.
  In particular, we look at the impact on minimal trap spaces,
  and identify a structural condition that guarantees their preservation.
  We examine the possible effects of variable elimination
  under three of the most popular approaches to control (attractor-based control,
  value propagation and control of minimal trap spaces),
  and under different update schemes (synchronous, asynchronous, generalized asynchronous).
  We provide some insights on the application of reduction, and an ample inventory of examples and counterexamples.
\end{abstract}

\section{Introduction}

In the investigation of complex systems where quantitative data is scarce,
one can resort to tracking only the absence or presence of the interacting entities.
Their interplay can be abstracted through logical rules,
resulting in the creation of a model usually known as a Boolean network \citep{schwab2020concepts,puvsnik2022review,kadelka2024meta}.
As models become larger and more elaborate, reduction techniques are adopted
to curb the model complexity (see e.g. \citet{naldi2009reduction,zanudo2013effective,veliz2014steady,argyris2023reducing}).
Among these, variable elimination is quite popular and natural.
It consists in the removal of an intermediate component, and a consequent rewiring
of the influence diagram to account for regulations that were mediated by this component.
The effects of such modifications are sometimes not intuitive, even in discrete dynamics.
For instance, while fixed points are always preserved,
the removal of a simple intermediate variable in a linear chain of variables of arbitrary length
can change the number of asynchronous cyclic attractors \citep{schwieger2024reduction}.
Here we make further investigations on the impact of elimination of variables on objects of interest
in the analysis of Boolean networks.
We look specifically at the effects on attractors, minimal trap spaces, and phenotype control strategies.

Attractors are often the first entities that are identified when building a model,
as they should capture the stable behaviours.
The consequences of variable elimination on attractors have been previously studied mostly
in the \emph{asynchronous} dynamics case \citep{naldi2009reduction,naldi2011dynamically,veliz2011reduction,schwieger2024reduction},
that is, under the update scheme where only one component can be updated in each transition.
Here we look also at other update choices.
The \emph{synchronous} update requires all changes to happen at the same time and, as we will see,
behaves more poorly than other updates with respect to variable elimination.
Besides the synchronous and asynchronous updates, we consider the \emph{general asynchronous} dynamics,
which allows the simultaneous update of any possible subset of the variables that can be updated,
and in particular contains all transitions of both the synchronous and asynchronous dynamics.
Even richer than the general asynchronous dynamics is the \emph{most permissive} dynamics,
which accounts for all possible behaviours that can be generated by multilevel versions of the Boolean network \citep{pauleve2020reconciling}.

Minimal trap spaces are interesting because they generally provide good approximations for attractors \citep{klarner2015approximating},
and at the same time are not as challenging to compute for Boolean biological models \citep{trinh2022minimal,moon2022computational}.
Under the most permissive semantics, attractors and minimal trap spaces coincide \citep{pauleve2020reconciling}.
Here we describe a simple structural property that guarantees preservation of minimal trap spaces (\cref{thm:min-ts}).
By ``structural'' we mean a condition on the interaction graph
which does not depend on the update chosen to generate the dynamics.
This particular condition requires that the variable being eliminated and its targets
have no regulators in common. When this condition is satisfied, we call the variable being eliminated a \emph{mediator}.

Identification of control strategies is one of the main objectives of
logical modelling of biological systems \citep{glass1973logical,zanudo2015cell,plaugher2023phenotype}.
Even in this rather niche context, control can be interpreted and approached in many ways,
e.g., by controlling nodes or edges, considering permanent, temporal or sequential interventions,
etc. (see for instance \citet{biane2018causal,sordo2020control,su2020sequential}).
Here we focus on \emph{phenotype control} achieved via permanent node interventions.
The objective is to find restrictions on the values of some variables
that are able to steer the dynamics towards a desired asymptotic behaviour.
We further distinguish between three type of interventions.
We consider attractor control strategies \citep{akutsu2012integer,zanudo2015cell,su2020dynamics,cifuentes2020control,cifuentes2022control}
that ensure that all attractors are contained in the desired phenotype;
a second type of control strategy, that guarantees that the minimal trap spaces are in the target phenotype \citep{pauleve2023reprogramming,riva2023tackling};
and a stronger class of interventions, which we call strategies by \emph{value propagation},
requiring that the fixed values propagate in the network until the phenotype variables are fixed \citep{samaga2010computing}.
Control strategies belonging to the latter category are probably the most popular,
for the following reasons: they are control strategies also under the other two definitions,
they apply independently of the update scheme, and can be identified quite efficiently
for example with Answer Set Programming \citep{kaminski2013minimal}.

After providing the formalization and notation required to address reduction and control in Boolean networks (\cref{sec:defs}),
we make preliminary observations about the effect of elimination of network components
on attractors and minimal trap spaces (\cref{sec:consequences}),
instrumental to the discussions about control in the last section (\cref{sec:control-reduction}).
We organise our investigations around two main questions: if a control strategy exists for a given phenotype in a Boolean network,
is a control strategy guaranteed to exist for a reduced version of the Boolean network?
And if a reduced Boolean network can be controlled for a given phenotype,
can we find a control intervention for the original network?
We consider the questions for all the aforementioned dynamics and control types,
for eliminated components that are mediators and in the general case.
We find that control strategies by value propagation are more robust to component elimination:
the first question can be answered always positively (\cref{thm:propagation}), and the second only partially
(\cref{ex:prop,ex:free-in-P-no-CS-to-CS,thm:perc-in-red,thm:perc-in-red-2}).
Removal of mediator nodes works well for control of minimal trap spaces (\cref{thm:min-ts-control}),
but is not a guarantee for good behaviour in the general attractor case,
as clarified by many counterexamples.

\section{Definitions and background}\label{sec:defs}

We set $\B = \{0, 1\}$.
Boolean networks on $n$ \emph{components} (or variables) are maps from $\B^n$ to itself,
used to model, for instance, the qualitative behaviour of interacting biological species \citep{schwab2020concepts,puvsnik2022review}.
$\B^n$ is called the \emph{state space} of networks on $n$ components.
We write $[n] = \{1,\dots, n\}$ for brevity.
The neighbour state of a state $x \in \B^n$ in direction $i \in [n]$ is denoted by $\bar{x}^i$.
Given a set $I \subseteq [n]$ and a state $x \in \B^n$, $x_I \in \B^I$ denotes the projection
of $x$ on the components in $I$.
For a set $A \subseteq \B^n$, $A_I$ denotes the set of states $x_I$ with $x \in A$,
and $f(A)$ is the image of $A$ under $f$ ($f(A)=\{f(x) \mid x\in A\}$).
Given a subset $A$ of $\B^{n-1}$, we write $A^{\star}$ for the largest subset of $\B^n$
that satisfies $A^{\star}_{[n-1]}=A$ (that is, $A^{\star} = \{x \in \B^n | x_{[n-1]} \in A\}$).

Consider a subset of $I$ of $[n]$ and a map $c \colon I \to \{0,1\}$.
The subset of $\B^n$ consisting of all states $x$ such that $x_i = c(i)$ for all $i \in I$
is called a \emph{subspace} of $\B^n$.
Variables in $I$ are said to be \emph{fixed} in the subspace,
while the other components are \emph{free}.
It is convenient to represent a subspace as an element of $\subs^n = \{0,1,{\star}\}^n$,
where ${\star}$ indicates that a component is free.
For example, the subspace $S = {\star} 01 \in \subs^3$ is the set $\{001, 101\}$, the first component is free ($S_1={\star}$),
and the second and third are fixed ($S_2=0$, $S_3=1$).
Note that, if $S \subseteq \B^{n-1}$ is a subspace, then $S^{\star}$ is also a subspace.

Dependencies between components as defined by their associated Boolean functions are captured by the
so-called \emph{interaction} or \emph{influence graph}.
This is a directed signed graph with set of nodes being the components $[n]$,
and admitting an edge from node $i$ to node $j$ of sign $s \in \{-1, 1\}$ if,
for some state $x \in \B^n$, $f_j(x) \neq f_j(\bar{x}^j)$, and $s = (f_j(\bar{x}^i) - f_j(x)) (\bar{x}^i_i - x_i)$.
In this case we say that $j$ is regulated by $i$. In case of $j=i$, $j$ is said to be \emph{autoregulated}.

In the following, the examples of Boolean networks are specified with propositional logic, using $\vee$ for \texttt{or}, while the symbol for \texttt{and} is omitted.

\subsection{Update schemes}

\begin{figure}
\begin{subfigure}{0.5\textwidth}
  \begin{minipage}{0.5\textwidth}
  \begin{equation*}
    \begin{tikzcd}[column sep=tiny]
      \circled{1} \arrow[loop left] \arrow[rr,bend right=20] & & \circled{2} \arrow[ll,-|,bend right=20] \arrow[dl,yshift=2pt,xshift=+2pt,bend right=20] \\
      & \circled{3} \arrow[ul,bend left=20] \arrow[ur,-|,yshift=-2pt,xshift=1pt,bend right=20]
    \end{tikzcd}
  \end{equation*}
  \end{minipage}%
  \begin{minipage}{4.7cm}
  \fbox{
  \resizebox{3.8cm}{!}{
  \begin{tikzpicture}
  \node[draw=lightgray] (000) at (  0,  0){000};
  \node[draw=lightgray] (100) at (\vd,  0){100};
  \node (010) at (  0,\vd){010};
  \node (110) at (\vd,\vd){110};

  \node (001) at (-\td,-\td){001};
  \node (101) at (\vd+\td,-\td){101};
  \node[draw=lightgray] (011) at ( -\td,\vd+\td){011};
  \node[draw=lightgray] (111) at (\vd+\td,\vd+\td){111};

  \path[->,draw,black]
  (001) edge (000)
  (010) edge (000)
  (010) edge (011)
  (011) edge (001)
  (100) edge (110)
  (101) edge (100)
  (110) edge (111)
  (110) edge (010)
  (111) edge (101)
  (010) edge[dashed] (001)
  (110) edge[dashed] (011)
  ;
  \end{tikzpicture}
  }
  }
  \end{minipage}%
  \caption{$f(x_1, x_2, x_3) = (x_1 \bar{x}_2 {\vee} x_1 x_3, x_1 \bar{x}_3, x_2)$}\label{fig:ex-all-a}
\end{subfigure}
\begin{subfigure}{0.5\textwidth}
  \begin{minipage}{0.6\textwidth}
  \begin{equation*}
    \begin{tikzcd}[column sep=small]
      \circled{1} \arrow[loop left] \arrow[r] & \circled{2} \arrow[loop right,-|]
    \end{tikzcd}
  \end{equation*}
  \end{minipage}%
  \begin{minipage}{0.5\textwidth}
  \fbox{
  \begin{tikzpicture}
  \node (00) at (0,0){00};
  \node (10) at (\vd,0){10};
  \node (01) at (0,\vd){01};
  \node (11) at (\vd,\vd){11};
  \path[->,draw,black]
  (01) edge (00)
  (10) edge[transform canvas={xshift=+1.6pt}] (11)
  (11) edge[transform canvas={xshift=-1.6pt}] (10)
  ;
  \end{tikzpicture}
  }
  \end{minipage}%
  \caption{$\rd{f}(x_1,x_2) = (x_1, x_1 \bar{x}_2)$}\label{fig:ex-all-b}
\end{subfigure}
\caption{(a) Interaction graph and state transition graphs of a Boolean network in 3 components.
States in boxes are representative states w.r.t. the component $n = 3$,
which is not autoregulated.
(b) Interaction graph and state transition graphs of the Boolean network obtained from the
network in (a) by elimination of component 3.
$\AD(\rd{f})$, $\SD(\rd{f})$ and $\GAD(\rd{f})$ coincide.
The transitions $110 \to 010$ and $110 \to 011$ are lost in the reduction.}\label{fig:ex-all}
\end{figure}
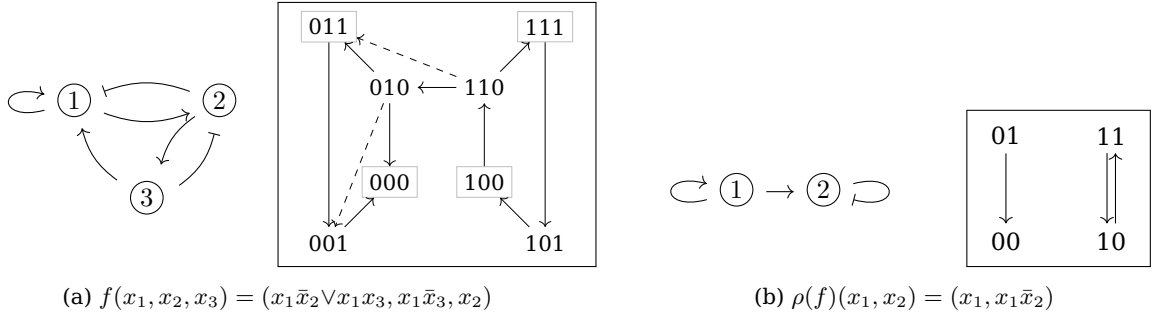

We define dynamics of a Boolean network $f$ on $n$ components as directed graphs with set of nodes
being the state space $\B^{n}$.
The edges, called \emph{transitions}, are defined depending on the update scheme as follows.

\begin{itemize}
  \item In the synchronous dynamics ($\SD(f)$) each state that is not fixed has exactly one successor, defined by its image under $f$,
        that is, the set of transitions is given by $\{ x\to y \mid x\neq y, y = f(x) \}$.
  \item In the asynchronous dynamics ($\AD(f)$) only transitions that involve the update of one component are considered:
        the set of transitions is $\{ x\to y\mid \exists i\in[n]: y=\bar x^i, y_i=f_i(x) \}$.
  \item The general asynchronous dynamics ($\GAD(f)$) allows for the update of any combination of possible components,
        and has therefore edge set $\{ x\to y\mid x\neq y,\forall i\in[n]: y_i\neq x_i\Rightarrow y_i=f_i(x) \}$.
\end{itemize}

Observe that all transitions in $\AD(f)$ and $\SD(f)$ are transitions in $\GAD(f)$.

Other definitions of dynamics are possible.
For instance, the most permissive dynamics cointains all transitions that are achievable
in a multivalued refinement of $f$ \citep{pauleve2020reconciling},
and contains in particular all transitions that are in $\GAD(f)$.
Although we do not consider the most permissive semantics explicitly here,
the results about control of minimal trap spaces have a bearing on most permissive dynamics,
because minimal trap spaces and attractors coincide in this case.

In the examples, we draw the transitions in asynchronous dynamics as normal arrows,
while the transitions found in synchronous dynamics are dashed (if not drawn as asynchronous),
and transitions in general asynchronous are dotted (if not drawn as asynchronous or synchronous).

\begin{example}
  \cref{fig:ex-all} (a) displays the interaction graph and the synchronous, asynchronous and general asynchronous
  state transition graphs of a Boolean network in 3 components.
  For instance, the state 100 has one successor (110) in all three dynamics,
  whereas the state 110 has one successor (011) in the synchronous,
  two successors (010 and 111) in the asynchronous, and three successors in the general asynchronous dynamics.
\end{example}

\subsection{Trap sets, trap spaces, attractors}

Given a state transition graph, a \emph{trap set} is a subset of the state space
that admits no outgoing transitions.

A trap set that is minimal with respect to inclusion is called an \emph{attractor}.
Attractors that consist of a single state are called \emph{fixed points} or \emph{steady states}.
Other attractors are called \emph{cyclic} or \emph{complex}.

A subspace that is also a trap set is called a \emph{trap space},
In other words, a subspace $T \in \subs^n$ is a trap space if $f(T) \subseteq T$,
that is, if $f_i(T) = T_i$ for all $i \in [n]$ such that $T_i \in \{0,1\}$.
A trap space $T$ is minimal if, given $T'$ trap space, $T' \subseteq T$
implies $T' = T$.

Fixed points and trap spaces are independent of the update scheme.
Minimal trap spaces are objects of particular interest.
By definition, each minimal trap space contains at least one attractor.
On the other hand, empirical studies of Boolean models of biological networks
found that minimal trap spaces are generally in one-to-one correspondence with attractors
of asynchronous dynamics \citep{klarner2015approximating}.
There are also classes of networks for which the one-to-one correspondence
between attractors and minimal trap spaces is guaranteed by structural properties
of the interaction graph \citep{naldi2023linear}.
Moreover, minimal trap spaces are exactly the attractors in most permissive dynamics.

\begin{example}
  The Boolean network in \cref{fig:ex-all-a} has four trap spaces:
  ${\star}{\star}{\star}$, $0{\star}{\star}$, $00{\star}$, $000$.
  There is only one minimal trap space, $000$, which is a fixed point,
  and there are no cyclic attractors in any dynamics.

  The network in \cref{fig:ex-all-b} has a fixed point ($00$) and a cyclic attractor ($\{10, 11\}$).
  They coincide with the minimal trap spaces.
\end{example}

\subsection{Reduction: elimination of components}\label{sec:reduction}

We recall the definition for elimination of non-autoregulated components \citep{naldi2009reduction,naldi2011dynamically,veliz2011reduction}.
For convenience and without loss of generality, we consider the elimination of the last component $n$.

Since $n$ is not autoregulated, for each $x \in \B^{n-1}$ the equality $f_n(x, 0) = f_n(x, 1)$ holds.
We call the state $(x, f_n(x,0))$ the \emph{representative state} of $\{(x, 0), (x, 1)\}$
for the elimination of component $n$.
It will also be convenient to refer to $(x, f_n(x,0))$ as the representative state of $x$.

The reduction $\rd{f} \colon \B^{n-1} \to \B^{n-1}$ of the Boolean network $f$
by elimination of component $n$ is then defined,
for each component $i \neq n$, as $f_i$ applied to the representative states:
for each $x \in \B^{n-1}$,
\begin{equation*}\label{eq:def-red}
  \rd{f}_i(x) = f_i(x, f_n(x, 0)) = f_i(x, f_n(x, 1)).
\end{equation*}

Equivalently, denoting $\sigma \colon \B^{n-1} \to \B^{n}$ the map that associates to each state $x$
the representative state of $x$, we can write
\begin{equation}\label{eq:def-red-2}
  \rd{f}_i = f_i \circ \sigma.
\end{equation}

\begin{figure}
\begin{subfigure}{0.32\textwidth}
  \begin{equation*}
    f =
    \begin{cases}
      f_1(x, x_n) \\
      f_2(x, x_n) \\
      \ \vdots \\
      f_{n-1}(x, x_n) \\
      f_n(x, 0) = f_n(x, 1)
    \end{cases}
  \end{equation*}

  $$\hspace{1cm} \big\downarrow \substack{\text{elimination} \\ \text{of } n}$$

  \begin{equation*}
    \rd{f} =
    \begin{cases}
      f_1(x, f_n(x, 0)) \\
      f_2(x, f_n(x, 0)) \\
      \ \vdots \\
      f_{n-1}(x, f_n(x, 0))
    \end{cases}
  \end{equation*}
\caption{}\label{fig:reduction-f}
\end{subfigure}
\begin{subfigure}{0.31\textwidth}
\centering
\begin{tikzpicture}
  \node (n) at (0,0){\circled{n}};
  \node (a) at (-1.3,1){\circled{\color{white}{a}}};
  \node (b) at (-1.2,0){\circled{\color{white}{a}}};
  \node (c) at (-1.3,-1){\circled{\color{white}{a}}};
  \node (d) at (1.3,0.6){\circled{\color{white}{a}}};
  \node (e) at (1.3,-0.6){\circled{\color{white}{a}}};
  \path[->,draw,black]
  (a) edge (n)
  (b) edge (n)
  (c) edge (n)
  (c) edge[loop left] (c)
  (n) edge (d)
  (n) edge (e)
  (d) edge[loop right] (d);
\end{tikzpicture}

\vspace{-0.2cm}

$$\hspace{1cm} \big\downarrow \substack{\text{elimination} \\ \text{of } n}$$

\vspace{0.3cm}

\begin{tikzpicture}
  \node (a) at (-1.3,1){\circled{\color{white}{a}}};
  \node (b) at (-1.2,0){\circled{\color{white}{a}}};
  \node (c) at (-1.3,-1){\circled{\color{white}{a}}};
  \node (d) at (1.3,0.6){\circled{\color{white}{a}}};
  \node (e) at (1.3,-0.6){\circled{\color{white}{a}}};
  \path (c) edge[loop left] (c);
  \path[->,draw,black,dashed]
  (a) edge (d)
  (b) edge (d)
  (c) edge (d)
  (a) edge (e)
  (b) edge (e)
  (c) edge (e)
  (d) edge[loop right] (d)
  ;
\end{tikzpicture}
\vspace{0.2cm}
\caption{}\label{fig:reduction-ig}
\end{subfigure}
\begin{subfigure}{0.4\textwidth}
  \begin{equation*}
  \resizebox{!}{1.5cm}{
  \begin{tikzcd}[row sep=small,column sep=small,ampersand replacement=\&]
    \& \& \& \overline{\sigma(x)}^j \\
    \& \& \sigma(x) \arrow[ur] \& \\
    \& \& \& (\bar{x}^j,x_n) \\
    (\bar{x}^i,x_n) \& \& (x,x_n) \arrow[uu,"n"'] \arrow[ll,"i"'] \arrow[ur,dash,dashed,gray,"j"] \&
  \end{tikzcd}}
  \end{equation*}

  $$\hspace{1cm} \big\downarrow \substack{\text{elimination} \\ \text{of } n}$$

  \vspace{0.8cm}

  \begin{equation*}
  \begin{tikzcd}[row sep=small,column sep=small]
    & & & \bar{x}^j \\
    \bar{x}^i & & x \arrow[ll,dash,dashed,gray,"i"'] \arrow[ur,"j"] &
  \end{tikzcd}
  \vspace{0.5cm}
  \end{equation*}
  \caption{}\label{fig:reduction-ad}
\end{subfigure}
  \caption{Schematics summarizing the idea behind elimination of a non-autoregulated component (component $n$ in the figure).
  (a) Effect on the update functions: all instances of $x_n$ are substituted with the update function $f_n$ of $n$.
  (b) Effect on the interaction graph: paths of length two that are mediated by $n$ become direct interactions
      or can disappear with the reduction.
  (c) Effect on the asynchronous dynamics: $\sigma(x)$ is the representative state of $(x, x_n)$.
      Only transitions that start from a representative state are guaranteed to be preserved.}\label{fig:reduction}
\end{figure}
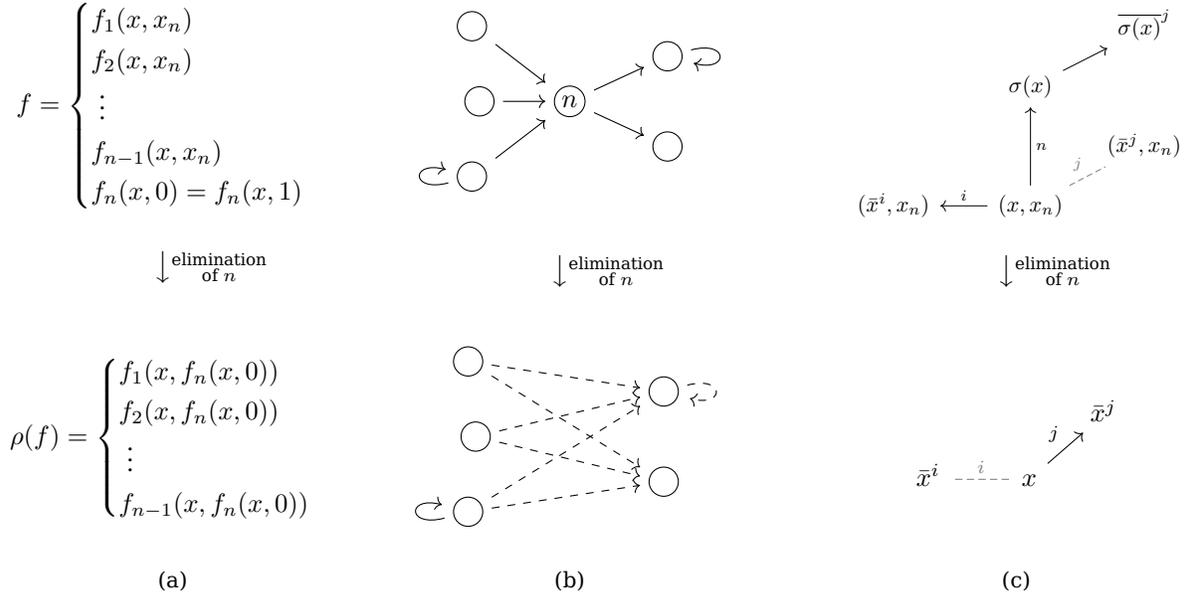

Intuitively, when the update function for component $n$ does not depend on $n$ itself,
one can replace all instances of $x_n$ in the update functions of other components with $f_n$,
obtaining a Boolean network that does not involve $n$ (\cref{fig:reduction-f}).
The relationships between the resulting dynamics and interaction graphs and the original
dynamics and interaction graphs have been studied in particular in \citep{naldi2009reduction,naldi2011dynamically,veliz2011reduction}.
In terms of regulatory structure, while interactions can disappear with the reduction (\cref{fig:reduction-ig})
the existence of a path of sign $s$ in the interaction graph of $\rd{f}$
implies the existence of a path of the same sign in the interaction graph of $f$.
Concerning the dynamics, one can easily derive the following:
\begin{itemize}
  \item[(1)] For all $x \in \B^n$, there is a transition from $\overline{\sigma(x)}^n$ to $\sigma(x)$
        in $\AD(f)$ and $\GAD(f)$ (but not necessarily in $\SD(f)$).
  \item[(2)] For $J \subseteq V \setminus \{n\}$ and $x \in \B^n$, for any dynamics $\D$,
        there exists a transition in $\D(\rd{f})$ from $x$ to $\bar{x}^J$
        if and only if there is a transition in $\D(f)$ from $\sigma(x)$ to $\overline{\sigma(x)}^J$.
  \item[(3)] As a consequence, $x \in \B^n$ is a fixed point for $\rd{f}$ if and only if
        $\sigma(x)$ is a fixed point for $f$, and there is a one-to-one mapping between
        the fixed points of $f$ and the fixed points of $\rd{f}$.
\end{itemize}
Looking at observation (2) we can state that a transition that starts at a non-representative state
is not represented in the reduced dynamics, unless a parallel transition exists that starts at its
corresponding representative state (\cref{fig:reduction-ad}).
Note how point (1) creates a difference between the synchronous dynamics and the other dynamics.
This distinction is at the source of many limitations that arise in the application of
elimination of components to synchronous dynamics.
We will later take a closer look at what happens to trap spaces (\cref{sec:min-ts}),
and discuss cyclic attractors (\cref{sec:attractors}).

\begin{example}
  In \cref{fig:ex-all-a}, the representative states for the elimination of component 3 are in boxes.
  For instance, since $f_3(110) = 1$, the representative state of $110$ and $111$ is the state $111$.
  The Boolean network resulting from the elimination is shown in \cref{fig:ex-all-b}.
  We can observe that the transition from $111$ to $101$ results in a transition from $11$ to $10$
  in the reduction ($111$ is a representative state), while the transitions
  from $110$ to $010$ or to $011$ are not preserved by the reduction, since no similar transitions
  exist with source the representative state $111$ of $110$.
\end{example}

\subsection{Control strategies}\label{sec:control}

In this work, a \emph{control strategy} to be applied on a Boolean network $f\colon\B^n \to \B^n$
is \emph{identified with a subspace} of $\B^n$.
Informally, the application of a control strategy consists in fixing the value
of a subset of the components.

The result of the application of control strategy $S$ to $f$
will be denoted by $C(f, S)$, and is defined as another Boolean network from $\B^n$ to itself.

For each component $i$, we set:
\begin{equation*}\label{def:control}
  C(f, S)_i = \begin{cases} f_i, & \textit{if $i$ is free in $S$},\\
                S_i, & \textit{otherwise.}
              \end{cases}
\end{equation*}

\begin{remark}\label{rmk:ig-subgraph}
  The interaction graph of $C(f, S)$ is a subgraph of the interaction graph of $f$.
\end{remark}

\begin{example}
  For the network in \cref{fig:order-a}, applying the control defined by $S={\star}1$ means
  changing the update function $f_2(x_1, x_2) = x_1$ to $C(f,S)_2(x_1,x_2) = 1$ (\cref{fig:order-b}).
\end{example}

One can observe that the elimination of a component and the application of a control strategy
commute, provided that the eliminated component is not fixed in the control strategy.
This is stated by the following proposition.

\begin{proposition}\label{prop:red-sub}
  Suppose that $n$ is free in $S$. Then $C(\rd{f}, S_{[n-1]}) = \rd{C(f, S)}$.
\end{proposition}
\begin{proof}
  If $i$ is fixed in $S$, then both $C(\rd{f}, S_{[n-1]})_i$ and $\rd{C(f, S)}_i$ equal $S_i$.
  If $i$ is free in $S$, then its update function is not changed by the application of the  control strategy,
  thus $C(\rd{f}, S_{[n-1]})_i = \rd{f}_i$ and $C(f, S)_i = f_i$.
  Hence, $C(f, S)_n = f_n$, and therefore
  $\rd{C(f, S)}_i = \rd{f}_i = C(\rd{f}, S_{[n-1]})_i$.
\end{proof}

Now, consider the removal of a component that is fixed in $S$.
The application of the control $S$ to $f$ and the elimination of the component,
when performed in a different order, can result in a different Boolean network.

For example, the restriction of $f(x_1,x_2) = (x_1 {\vee} x_2, x_1)$ to $S = {\star} 1$
gives $C(f, S) = (x_1 {\vee} x_2, 1)$ and $\rd{C(f, S)}(x_1) = 1$,
whereas $\rd{f}(x_1) = x_1 = C(\rd{f}, S_{[1]}={\star})(x_1)$ (see \cref{fig:order}).

\begin{figure}
\begin{subfigure}{0.3\textwidth}
\centering
\fbox{
\begin{tikzpicture}
\filldraw[fill=black!20!white, draw=black, opacity=0.3] (-0.35*\vd,0.7*\vd) rectangle (1.3*\vd,1.3*\vd);
\node (00) at (0,0){00};
\node (10) at (\vd,0){10};
\node (01) at (0,\vd){01};
\node (11) at (\vd,\vd){11};
\path[->,draw,black]
(01) edge (00)
(01) edge (11)
(10) edge (11)
(01) edge[dashed] (10)
;
\node at (1.5*\vd,1.0*\vd){$S$};
\end{tikzpicture}}
\caption{$f(x_1,x_2) = (x_1 {\vee} x_2, x_1)$}\label{fig:order-a}
\end{subfigure}%
\begin{subfigure}{0.3\textwidth}
\centering
\fbox{
\begin{tikzpicture}
\filldraw[fill=black!20!white, draw=black, opacity=0.3] (-0.35*\vd,0.7*\vd) rectangle (1.3*\vd,1.3*\vd);
\node (00) at (0,0){00};
\node (10) at (\vd,0){10};
\node (01) at (0,\vd){01};
\node (11) at (\vd,\vd){11};
\path[->,draw,black]
(00) edge (01)
(01) edge (11)
(10) edge (11)
;
\node at (1.5*\vd,1.0*\vd){$S$};
\end{tikzpicture}}
\caption{$C(f, S)(x_1, x_2) = (x_1 {\vee} x_2, 1)$}\label{fig:order-b}
\end{subfigure}%
\begin{subfigure}{0.2\textwidth}
\centering
  \fbox{
  \begin{tikzpicture}
  \node (0) at (0,0){0};
  \node (1) at (\vd,0){1};
  \end{tikzpicture}}
\caption{$\rd{f}(x_1) = x_1$}\label{fig:order-c}
\end{subfigure}%
\begin{subfigure}{0.2\textwidth}
\centering
  \fbox{
  \begin{tikzpicture}
  \node (0) at (0,0){0};
  \node (1) at (\vd,0){1};
  \path[->,draw,black]
  (0) edge (1);
  \end{tikzpicture}}
\caption{$\rd{C(f, S)}(x_1)=1$}\label{fig:order-d}
\end{subfigure}
\caption{Example illustrating that, if $S_n \neq \star$, then $C(\rd{f},S_{[n-1]})$ and $\rd{C(f,S)}$ can differ.}\label{fig:order}
\end{figure}
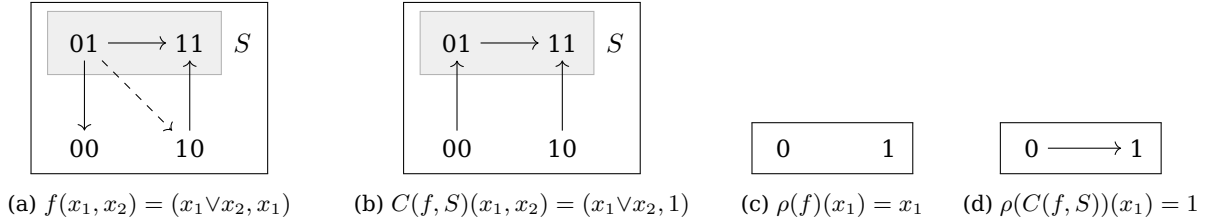

In light of this latter remark, \emph{we restrict the analysis of control strategy behaviour under
reduction to the removal of components that are free in the control strategy:}
\begin{equation}\label{eq:assumption-Sn-free}
  S_n = \star.\tag{A}
\end{equation}

\subsubsection{Phenotype control}

The objective of control is typically the redirection of the asymptotic behaviour
towards a phenotype of interest.
For the purpose of this work, a \emph{phenotype} is defined as a subspace,
i.e., it is identified by fixing some components to specific values.

We can think of components that are fixed in a phenotype as \emph{readouts}
of the model; phenotypes are often defined using output components
(components that are not regulators of any other component).
Control strategies instead work on components that can be modified, and
control often focuses on (but is not necessarily limited to) input components,
meaning components that are not the target of any regulator.
Since components that are fixed in phenotypes or in control strategies fulfill two opposite roles,
it is reasonable to \emph{limit the control strategies under consideration
to subspaces $S$ that do not fix any component that is fixed in the target phenotype $P$}:
\begin{equation}\label{eq:assumption-Pi-Si}
   i \in [n], \ P_i \in \{0,1\} \ \Rightarrow \ S_i = \star.\tag{B}
\end{equation}
Assumption \ref{eq:assumption-Pi-Si} gives a restriction on the possible control strategies that can be investigated for
a given phenotype, adding to assumption \ref{eq:assumption-Sn-free}, which imposes
that components being eliminated must be free in the control strategy.
Note that we do not make any restriction on $P_n$, that is, we do not assume that
the eliminated component is free in the phenotype. In fact, we will analyse the
two cases separately ($n$ free in $P$ and $n$ fixed in $P$).
In both of these cases, as per assumption \ref{eq:assumption-Sn-free},
$n$ is not allowed to be targeted by the control strategy.

We distinguish between three types of phenotype control here (see \cref{fig:control}).
The first looks at ensuring that all attractors are in the desired phenotype,
and depends on the dynamics of interest.

Fix a Boolean network $f$ on $n$ variables and a phenotype $P \in \subs^n$.

\begin{definition}(\emph{Phenotype control for attractors})
  A subspace $S$ is an attractor-control strategy for $(f, P)$
  under dynamics $\dyn$ if all the attractors of the dynamics $\dyn(C(f,S))$ are contained in $P$.
\end{definition}

A different approach focuses on controlling minimal trap spaces only,
and is therefore independent of the dynamics.

\begin{definition}(\emph{Phenotype control for minimal trap spaces})
  A subspace $S$ is an MTS-control strategy for $(f, P)$
  if all the minimal trap spaces of $C(f, S)$ are contained in $P$.
\end{definition}

Control of minimal trap spaces is neither strictly stronger nor strictly weaker
than attractor control, as illustrated by the following examples.
In figures, the gray boxes cover states belonging to the target phenotype.

\begin{example}\label{ex:attrs-no-mts}
  (Attractor-control strategy that is not an MTS-control strategy)
  Consider the asynchronous dynamics for the Boolean network
  $f(x_1,x_2,x_3) = (x_2 \bar{x}_3, x_3 \bar{x}_2, x_2 {\vee} \bar{x}_3)$ (\cref{fig:ex:attrs-no-mts-a}).
  Take $P = 0{\star}{\star}$. Since the unique attractor of $\AD(f)$ ($\{000, 001, 011\}$) is contained in $P$,
  $S = {\star}{\star}{\star}$ is an attractor-control strategy for $(f, P)$.
  However, $f$ admits only one trap space, the full state space.
  Hence, $S$ is not an MTS-control strategy for $(f, P)$.
  Similarly,
  $S= {\star}{\star}{\star}$ is an attractor-control strategy for the synchronous dynamics of
  $f(x_1, x_2, x_3) = (x_2 x_3, x_3, \bar x_3)$, with $P = 0{\star}{\star}$ (\cref{fig:ex:attrs-no-mts-b}),
  $S= {\star}{\star}{\star}{\star}$ is an attractor-control strategy for the general asynchronous dynamics of
  $f(x_1, x_2, x_3, x_4) = (x_2 x_3 x_4, x_4(x_2 {\vee} \bar x_1 \bar x_3), \bar x_1(x_2 x_3 {\vee} \bar x_2 \bar x_4), x_3 \bar x_1)$
  and the phenotype $P=0{\star}{\star}{\star}$ (graph not shown).

  \begin{figure}
  \begin{subfigure}{0.33\linewidth}
    \centering
    \resizebox{\linewidth}{!}{
    \fbox{
    \begin{tikzpicture}
    \filldraw[fill=black!60!white, draw=black, opacity=0.3] (-1.2*\vd,-\vd) rectangle (0.4*\vd,2*\vd);
    \node (000) at (  0,  0){000};
    \node (100) at (\vd,  0){100};
    \node (010) at (  0,\vd){010};
    \node (110) at (\vd,\vd){110};

    \node (001) at (-\td,-\td){001};
    \node (101) at (\vd+\td,-\td){101};
    \node (011) at ( -\td,\vd+\td){011};
    \node (111) at (\vd+\td,\vd+\td){111};

    \path[->,draw,black]
    (000) edge[transform canvas={xshift=+4pt}] (001)
    (001) edge (000)
    (001) edge (011)
    (010) edge (000)
    (010) edge (110)
    (010) edge (011)
    (011) edge[transform canvas={xshift=+3pt}] (001)
    (100) edge (000)
    (100) edge[transform canvas={xshift=+4pt}] (101)
    (101) edge (100)
    (101) edge (001)
    (101) edge[transform canvas={xshift=+3pt}] (111)
    (110) edge (100)
    (110) edge (111)
    (111) edge (101)
    (111) edge (011)
    ;
    \end{tikzpicture}
    }
    }
    \caption{$(x_2 \bar{x}_3, x_3 \bar{x}_2, x_2 {\vee} \bar{x}_3)$}\label{fig:ex:attrs-no-mts-a}
  \end{subfigure}
  \begin{subfigure}{0.33\linewidth}
    \centering
    \resizebox{\linewidth}{!}{
    \fbox{
    \begin{tikzpicture}
    \filldraw[fill=black!60!white, draw=black, opacity=0.3] (-1.2*\vd,-\vd) rectangle (0.4*\vd,2*\vd);
    \node (000) at (  0,  0){000};
    \node (100) at (\vd,  0){100};
    \node (010) at (  0,\vd){010};
    \node (110) at (\vd,\vd){110};

    \node (001) at (-\td,-\td){001};
    \node (101) at (\vd+\td,-\td){101};
    \node (011) at ( -\td,\vd+\td){011};
    \node (111) at (\vd+\td,\vd+\td){111};

    \path[->,draw,black]
    (000) edge[dashed] (001)
    (001) edge[dashed,transform canvas={xshift=-2pt}] (010)
    (010) edge[dashed,transform canvas={xshift=2pt}] (001)
    (011) edge[dashed] (110)
    (110) edge[dashed,bend left=15] (001)
    (100) edge[dashed] (001)
    (101) edge[dashed] (010)
    (111) edge[dashed] (110)
    ;
    \end{tikzpicture}
    }
    }
    \caption{$(x_2 x_3, x_3, \bar x_3)$}\label{fig:ex:attrs-no-mts-b}
  \end{subfigure}%
  \begin{subfigure}{0.33\linewidth}
  \resizebox{\linewidth}{!}{
  \fbox{
  \begin{tikzpicture}
  \filldraw[fill=black!60!white, draw=black, opacity=0.3] (-1.2*\vd,-\vd) rectangle (0.4*\vd,2*\vd);
  \node (000) at (  0,  0){000};
  \node (100) at (\vd,  0){100};
  \node (010) at (  0,\vd){010};
  \node (110) at (\vd,\vd){110};

  \node (001) at (-\td,-\td){001};
  \node (101) at (\vd+\td,-\td){101};
  \node (011) at ( -\td,\vd+\td){011};
  \node (111) at (\vd+\td,\vd+\td){111};

  \path[->,draw,black]
  (001) edge (011)
  (010) edge (110)
  (011) edge (010)
  (100) edge (101)
  (101) edge (001)
  (110) edge (100)
  (111) edge (101)
  (111) edge (110)
  (111) edge (011)
  (111) edge[dashed,bend left=10] (000)
  (111) edge[dotted,bend left=10] (100)
  (111) edge[dotted,bend left=15] (001)
  (111) edge[dotted,bend right=10] (010)
  ;
  \end{tikzpicture}
  }
  }
    \caption{$((x_1{\vee}x_2) \bar{x}_3, \bar{x}_1 (x_2{\vee}x_3), \bar{x}_2(x_1{\vee}x_3))$}\label{fig:ex:mts-no-attrs}
  \end{subfigure}
  \caption{$S = \B^3$ is an attractor-control strategy for $P=0{\star}{\star}$ for an asynchronous dynamics (case (a)),
  for a synchronous dynamics (case (b)). On the other hand, $S$ is not an MTS-control strategy.
  (c): $S = \B^3$ is an MTS-control strategy for $P=0{\star}{\star}$, since the unique minimal trap space is the fixed point $000$.
  $S$ is not an attractor-control strategy in any of the three dynamics, given the existence of the attractor $\{001, 010, 011, 100, 101, 110\}$.}\label{fig:ex:attrs-no-mts}
  \end{figure}
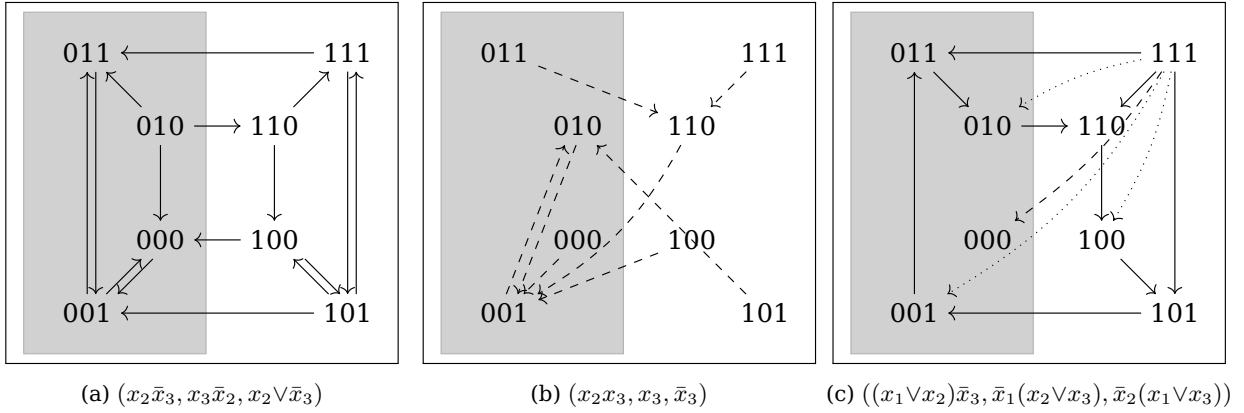
\end{example}

\begin{example}\label{ex:mts-no-attrs}
  Since attractors can exist outside of minimal trap spaces,
  MTS-control strategies are not necessarily attractor-control strategies.
  An example of such situation is given in \cref{fig:ex:mts-no-attrs}.
\end{example}

It should be noted that there are situations where MTS-control strategies are guaranteed to be also attractor-control strategies.
This is the case for instance for asynchronous dynamics of networks that admit a \emph{linear cut} \citep{naldi2023linear},
for which all asynchronous attractors are contained in minimal trap spaces.
Control of minimal trap spaces also translates to attractor control for most permissive dynamics \citep{pauleve2020reconciling}.

To discuss a third phenotype control scenario, we need an additional definition.

We call \emph{propagation} (or \emph{percolation}) function for $f$
the map $\Phi_f \colon \subs^n \to \subs^n$ that associates to each subspace $S \in \subs^n$
the minimal subspace, under inclusion, that contains $f(S)$.

Note that, if $S \in \subs^n$ is a trap space, $\Phi_f(S)$ is also a trap space, and $f(S) \subseteq S$.
Therefore, in this case there exists $k \leq n$ such that $\Phi_f^k(S) = \Phi_f^{k+i}(S)$ for all $i \in \mathbb{N}$.
We write $\phi(f) = \Phi_f^n(\B^n)$ for simplicity.

\begin{definition}(\emph{Phenotype control by value propagation})
  A subspace $S$ is a control strategy \emph{by (value) propagation} for $(f, P)$
  if $\phi(C(f, S))$ is contained in $P$.
\end{definition}

$S$ is a control strategy by propagation if fixing the components as defined by $S$
induces other components to get fixed under $f$ and so forth, until all the components fixed in the phenotype $P$
are fixed to their value in $P$.
Clearly all minimal trap spaces and all attractors of $f$, in any dynamics,
are contained in $\phi(f)$.
As a consequence, a control strategy by propagation is an attractor-control strategy in any dynamics,
and an MTS-control strategy. The converse is not true.

\begin{example}\label{ex:attrs-prop}
  The control strategies given in \cref{ex:attrs-no-mts} are attractor-control strategies
  but not control strategies by value propagation.
  For the example in \cref{fig:ex:mts-no-attrs}, ${\star}{\star}{\star}$ is an MTS-control strategy
  and not a control strategy by value propagation.
  For the Boolean network in \cref{fig:ex-all} (a),
  the full space $S = {\star}{\star}{\star}$ is an attractor-control strategy under all dynamics
  and an MTS-control strategy for $(f, P)$ with $P = 0{\star}{\star}$, but not a control strategy by value propagation.
\end{example}

Control strategies by value propagation have the desirable property of
working independently of the dynamics considered, as happens for MTS-control strategies.
Control strategies by value propagation can be thought of as particularly ``robust'' since they
allow control of all attractors in all updates.

\begin{figure}
  \centering
  \newcommand{\bdist}{5cm}
  \begin{tikzpicture}
    \node[rectangle,draw,minimum width=4,minimum height=2,line width=0.3mm] (prop) at (0,-0.2*\bdist) {\begin{tabular}{c}control by\\value propagation\end{tabular}};
    \node[rectangle,draw,minimum width=4,minimum height=2,line width=0.3mm] (attr) at (\bdist,0) {\begin{tabular}{c}attractor\\control\end{tabular}};
    \node[rectangle,draw,minimum width=4,minimum height=2,line width=0.3mm] (mints) at (\bdist,-0.4*\bdist) {\begin{tabular}{c}control of\\minimal trap spaces\end{tabular}};

    \draw[-{Implies},double,thick] ([shift={(0.1cm,+0.2cm)}]prop.east) -- ([shift={(-0.1cm,+0.2cm)}]attr.west);
    \draw[notImplies,double,thick,lightgray] ([shift={(-0.1cm,0.0cm)}]attr.west) -- ([shift={(0.1cm,0.0cm)}]prop.east) node[midway,below,shift={(0cm,-0.1cm)}]{\scriptsize ex.~\ref{ex:attrs-prop}};
    \draw[-{Implies},double,thick] ([shift={(0.1cm,-0.1cm)}]prop.east) -- ([shift={(-0.1cm,-0.1cm)}]mints.west);
    \draw[notImplies,double,thick,lightgray] ([shift={(-0.1cm,-0.3cm)}]mints.west) -- ([shift={(0.1cm,-0.3cm)}]prop.east) node[midway,left,shift={(0cm,-0.1cm)}]{\scriptsize ex.~\ref{ex:attrs-prop}};
    \draw[notImplies,double,thick,lightgray] ([shift={(-0.2cm,-0.1cm)}]attr.south) -- ([shift={(-0.2cm,+0.1cm)}]mints.north) node[midway,left,shift={(0cm,-0.1cm)}]{\scriptsize ex.~\ref{ex:attrs-no-mts}};
    \draw[notImplies,double,thick,lightgray] ([shift={(+0.2cm,0.1cm)}]mints.north) -- ([shift={(+0.2cm,-0.1cm)}]attr.south) node[midway,right,shift={(0cm,-0.1cm)}]{\scriptsize ex.~\ref{ex:mts-no-attrs}};
  \end{tikzpicture}

  \caption{Relationship between the three different approaches to phenotype control studied in this
  paper. The black double-lined arrows indicate total inclusion of phenotype control: any control by
  value propagation is an attractor-control and MTS-control strategy.
  Gray double-lined arrows with a slash indicate that the relationship is not always true.
  A reference to a counterexample is provided.}\label{fig:control}
\end{figure}
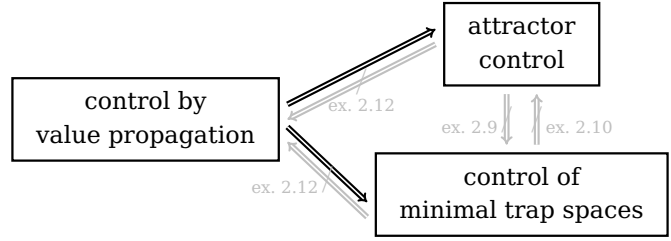

\section{Consequences of reduction on asymptotic dynamics}\label{sec:consequences}

It is well known that elimination of components affects the asymptotic dynamics of Boolean networks.
The map described in \cref{fig:ex-all} shows an example of reduction having an impact on the number of
minimal trap spaces and the number of attractors in all update modes.
In this section we first consider the effect of component elimination on minimal trap spaces,
and identify a structural condition for their preservation:
elimination of \emph{mediator} components, i.e., components having
regulators distinct from the regulators of their targets.
Then we clarify some differences and commonalities on the effects of reduction
on the asymptotic behaviour under different update schemes.

\subsection{Minimal trap spaces}\label{sec:min-ts}

We first list some general observations about trap spaces and elimination of components.

\begin{proposition}\label{prop:trap-spaces}
  Suppose that $T \in \subs^n$ is a trap space for $f$. Then:
  \begin{itemize}
    \item[(i)] $T_{[n-1]}$ is a trap space for $\rd{f}$.
    \item[(ii)] if $T$ is a minimal trap space and $T_n \in \{0, 1\}$, then $T_{[n-1]}$ is a minimal trap space.
    \item[(iii)] if $T$ is a minimal trap space and $T_i \in \{0, 1\}$ for all targets $i$ of $n$,
      then $T_{[n-1]}$ is a minimal trap space.
  \end{itemize}
\end{proposition}
\begin{proof}
  (i) For all $x \in T_{[n-1]}$ and for all $i \neq n$, if $T_i$ is in $\{0,1\}$,
  then by \cref{eq:def-red-2}, since $\sigma(x)$ is in $T$, we have $\rd{f}_i(x) = f_i(\sigma(x)) = T_i$.

   \medskip

   (ii) $T_{[n-1]}$ is a trap space by the first point.
   Suppose that $T' \subseteq T_{[n-1]}$ is a trap space.
   Take $i \neq n$ such that $T'_i$ is in $\{0,1\}$, we want to show that $T_i = T'_i$.
   For any state $x \in T$ we have $f_i(x) = f(x_{[n-1]}, T_n) = f(x_{[n-1]}, f_n(x)) = \rd{f}_i(x_{[n-1]}) = T'_i$.

   \medskip

   (iii)
   Suppose that $T' \subseteq T_{[n-1]}$ is a trap space.
   Take $i \neq n$ such that $T'_i$ is in $\{0,1\}$, we want to show that $T_i = T'_i$.
   For any state $x \in T$:
   \begin{itemize}
     \item if $i$ is not a target of $n$, then $f_i(x) = f_i(\bar{x}^n)$, therefore $f_i(x) = f_i(\sigma(x_{[n-1]})) = \rd{f}_i(x_{[n-1]}) = T'_i$;
     \item if $i$ is a target of $n$, then, since $T_i$ is in $\{0,1\}$ and representative states of states in $T$ are in $T$, we have $T_i = f_i(x) = f_i(\sigma(x_{[n-1]})) = \rd{f}_i(x_{[n-1]}) = T'_i$.
   \end{itemize}
\end{proof}

For each minimal trap space $T$ of $f$, the reduction $\rd{f}$ admits at least one minimal trap space
contained in $T_{[n-1]}$. The reduced network can also admit ``new'' trap spaces outside of
projections of minimal trap spaces of $f$.
We introduce some terminology to relate the set of minimal trap spaces of a network
to the set of minimal trap spaces of its reduction.

\begin{definition}\label{def:pres-min-ts}
  The minimal trap spaces of $f$ are \emph{strictly preserved} by the reduction if,
  for each minimal trap space $T'$ of $\rd{f}$ there exists a minimal trap space $T$ of $f$
  such that $T' = T_{[n-1]}$.
\end{definition}

The form of preservation introduced by the definition is rather strong.
If the minimal trap spaces are strictly preserved by the reduction, it is easy to see that,
given $T$ minimal trap space for $f$, $T_{[n-1]}$ is a minimal trap space for $\rd{f}$.
Therefore the minimal trap spaces of $f$ and $\rd{f}$ are in one-to-one correspondence.

The following result gives a sufficient condition for the preservation of minimal trap spaces.

\begin{theorem}\label{thm:min-ts}
  Suppose that no regulator of $n$ regulates a target of $n$.
  Then the minimal trap spaces of $f$ are strictly preserved by the elimination of $n$.
\end{theorem}
\begin{proof}
  Consider $T'$ minimal trap space for $\rd{f}$.
  Suppose that there is no minimal trap space $T$ for $f$ such that $T' = T_{[n-1]}$.
  We show that there exists a regulator $j$ of $n$ that regulates a target $i$ of $n$.

  Write $T$ for the smallest trap space for $f$ that satisfies $T' \subseteq T_{[n-1]}$.

  If $T' = T_{[n-1]}$, then, by hypothesis, $T$ is not a minimal trap space.
  That is, $T$ contains a smaller trap space $T''$.
  By definition of $T$, $T''_{[n-1]}$ does not contain $T'_{[n-1]}$.
  By \cref{prop:trap-spaces} $T''_{[n-1]}$ is therefore a trap space for $\rd{f}$
  strictly contained in $T'$, in contradiction with the minimality of $T'$.
  Hence $T' \neq T_{[n-1]}$.

  Now consider the subspace $S \in \subs^n$ that satisfies $S_i = T'_i$ for $i \neq n$ and $S_n = T_n$.

  Suppose that $T_n$ is in $\{0,1\}$.
  Then, since $S$ is contained in $T$, for any $x \in S$ we have $f_n(x) = S_n = T_n$.
  Therefore for $i$ fixed in $T'$ we have $f_i(x) = f_i(x_{[n-1]}, T_n) = \rd{f}_i(x_{[n-1]})$,
  and since $x_{[n-1]}$ is in $T'$ we find $f_i(x) = T'_i = S_i$.
  Therefore $S$ is a trap space that satisfies $T' \subseteq S_{[n-1]}$ strictly contained in $T$,
  contradicting the definition of $T$.

  We therefore have that $S_n = T_n = {\star}$.
  $S$ is strictly contained in $T$, and is not a trap space by definition of $T$.
  Therefore there exists a component $i$ that is fixed in $S$
  such that $f_i$ is not constantly equal to $S_i$ on $S$.
  Since $T'$ is a trap space for $\rd{f}$ and $i$ is fixed in $T'$,
  we have that $\rd{f}_i(x_{[n-1]}) = f_i(\sigma(x_{[n-1]})) = S_i$ for all $x \in S$.

  Now suppose that $f_n$ is constant on $S$ and equal to $a$.
  Consider the subspace $S' = S \cap \{x_n = a\}$ which is contained in $S$.
  Then $f_n(x) = a = x_n$ for all $x \in S'$, and for all $j$ fixed in $S$
  we have $f_j(x) = f_j(x_{[n-1]}, a) = f_j(\sigma(x_{[n-1]})) = S_j$,
  and $S'$ is a trap space strictly contained in $T$ that satisfies $T' \subseteq S'_{[n-1]}$,
  a contradiction.

  We can therefore apply \cref{lemma:regulators} to $f$ and the subspace $S$.
  Since $i$ is fixed in $S$, the lemma gives the existence of a component $j \neq i$
  that regulates both $n$ and $i$.
\end{proof}

\begin{lemma}\label{lemma:regulators}
  Suppose that there exists a subspace $S$ with $S_n = {\star}$ such that
  $f_i$ and $f_n$ are not constant on $S$ and $\rd{f}_i$ is constant on $S_{[n-1]}$.
  Then there exists a component $j \neq n$ such that $S_j = {\star}$ that is a regulator
  of both $n$ and $i$.
\end{lemma}
\begin{proof}
  If $i$ does not depend on $n$ on $S$, then for all $x \in S$ we have
  $f_i(x) = \rd{f}_i(x_{[n-1]})$ and $f_i$ is constant on $S$, contradicting the hypothesis.
  Therefore, $i$ is a target of $n$, and there exists a state $w \in S$ such that
  $w \neq \sigma(w_{[n-1]})$, $f_i(\sigma(w_{[n-1]})) = a$ and $f_i(w) = 1 - a$.
  Set $b = w_n$, so that $f_n(w) = 1 - b$.

  Since $f_n$ is not constant on $S$, there exists a state $y \in S$ such that $f_n(y) = b$.
  We can assume $\sigma(y_{[n-1]}) = y$, that is, $y_n = b$.
  Because $y$ is a representative state, we have $f_i(y) = a$.

  Summarizing, we have
  \[w_n = y_n = b,\]
  \[f_n(w) = 1 - b, f_n(y) = b,\]
  \[f_i(w) = 1 - a, f_i(y) = a.\]

  Observe that $w$ and $y$ are different states.
  Take the closest pair of states $w, y$ in $S$ that satisfy these conditions.

  Take a neighbour $z = \bar{w}^j$ of $w$ in $S$ closer to $y$ than $w$ ($z$ might coincide with $y$).
  Observe that $z_n = w_n = y_n = b$, and $j \neq n$.

  If $f_n(z) = b \neq f_n(w)$ (in particular, $z$ is representative), then, by the hypothesis on $\rd{f}_i = f_i \circ \sigma$,
  $f_i(z) = a \neq f_i(w)$, and $j$ is a regulator of both $i$ and $n$.

  If instead $f_n(z) = 1 - b$, then $z \neq y$ and, since the distance from $w$ to $y$ is minimal,
  again we must have $f_i(z) = a \neq f_i(w)$, and $j$ is a regulator of $i$.
  Now consider $v = \bar{y}^j$.
  If $f_n(v) = b$, then $f_i(v)=a$, contradicting the minimality of the distance between $w$ and $y$.
  Therefore, $f_n(v) = 1 - b \neq f_n(y)$ and $j$ regulates $n$, which concludes.
\end{proof}

The theorem gives a simple structural condition for selecting components to
eliminate without affecting the minimal trap spaces.
We have shown in particular that, if $T$ is a minimal trap space for $f$ and
$T'$ is a minimal trap space for $\rd{f}$ strictly contained in $T_{[n-1]}$, then
any component $i$ that is fixed in $T'$ and not in $T$
is regulated by $n$, as well as by at least one regulator of $n$ distinct from $i$.

We say that a component is \emph{linear} if it has exactly one regulator and one target.
A linear mediator component is therefore a particularly simple intermediate
whose role is just to mediate the regulation between two components.
In investigating mediator components, we considered whether an added assumption of linearity might guarantee better results
in terms of preservation of control strategies than the more general mediator assumption,
and found no additional benefits.
In the examples that investigate the impact of removal of mediator components,
we will consider in particular the elimination of linear mediator components,
in order to illustrate that even a seemingly minor modification of the interaction graph
can have consequences on the controllability of a network.

\subsection{Attractors}\label{sec:attractors}

Contrary to trap spaces, attractors are dependent of the update scheme.
The impact of reduction on attractors has been studied mostly
under asynchronous dynamics \citep{naldi2009reduction,naldi2011dynamically,veliz2011reduction,tonello2023attractor,schwieger2024reduction}.
Here we make some observations that highlight some differences
between synchronous dynamics and other updating schemes.

We first observe that trap sets are converted to trap sets in the reduction, except
in the synchronous dynamics.

\begin{lemma}\label{lemma:trap-sets}
  For $\D$ in $\{\AD, \GAD\}$,
  if $A \subseteq \B^n$ is a trap set for $\D(f)$,
  then $A_{[n-1]}$ is a trap set for $\D(\rd{f})$.
  Moreover, $\sigma(x)$ is in $A$ for all $x \in A_{[n-1]}$.
\end{lemma}
\begin{proof}
  The last observation follows from the fact that, for $x \in A_{[n-1]}$,
  either $\sigma(x)$ or $\overline{\sigma(x)}^n$ is in $A$,
  and the first is a successor of the second.

  Suppose that $x \in A_{[n-1]}$ and that $\D(\rd{f})$
  contains a transition from $x$ to $y \neq x$.
  We want to show that $y$ is in $A_{[n-1]}$.
  Call $I$ the set of indices such that $\bar{x}^I = y$.

  For all $i \in I$, $\rd{f}_i(x) = f_i(\sigma(x)) \neq x_i$,
  therefore there is a transition in $\D(f)$ from $\sigma(x)$ to $\overline{\sigma(x)}^I$.
  Since the representative state $\sigma(x)$ belongs to $A$,
  $\overline{\sigma(x)}^I$ is also in $A$,
  and $\overline{\sigma(x)}^I_{[n-1]} = \bar{x}^I = y$ is in $A_{[n-1]}$.
\end{proof}

\begin{example}
  Take $f(x_1,x_2,x_3) = (x_1 \bar{x}_3, 0, 0)$, which reduces to $\rd{f}(x_1,x_2) = (x_1, 0)$
  under elimination of $x_3$.
  The states $000$ and $111$ both map to $000$, hence $A = \{000, 111\}$ is a trap set for $\SD(f)$.

  In $\SD(\rd{f})$, $00$ maps to $00$ but $11$ maps to $10$,
  hence $A_{[n-1]} = \{00, 11\}$ is not a trap set.
\end{example}

\begin{lemma}\label{lemma:no-2-attrs}
  For $\D$ in $\{\AD, \GAD\}$,
  if $A$ is an attractor for $\D(f)$, then $A^{\star}_{[n-1]}$
  contains at most one attractor for $\D(f)$.
\end{lemma}
\begin{proof}
  Take a state $x$ in $A_{[n-1]}$. Then either $\overline{\sigma(x)}^n$ or $\sigma(x)$ belongs to $A$.
  Since in $\D(f)$ there is a transition from $\overline{\sigma(x)}^n$ to $\sigma(x)$,
  if $\overline{\sigma(x)}^n$ belongs to an attractor, then $\sigma(x)$ belongs to the same attractor.
\end{proof}

Consequence of \cref{lemma:trap-sets,lemma:no-2-attrs} is that
in asynchronous and generalized asynchronous dynamics
the number of attractors cannot decrease with the reduction.

As happens for \cref{lemma:trap-sets}, \cref{lemma:no-2-attrs} also fails for the synchronous dynamics,
since a state and its representative are not always linked by a transition.

\begin{example}\label{ex:synch}
  Consider the map $f(x_1,x_2) = (x_2, x_1)$ and the elimination of the second component.
  In the synchronous dynamics, there is no transition from state $01$ to its representative $00$,
  and from state $10$ to its representative $11$.

  The dynamics has three attractors: the steady states $00$ and $11$, and the cycle $A = \{01, 10\}$.
  The cycle projects to $A_{[n-1]} = \{0,1\}$, and $A^{\star}_{[n-1]}$ contains three attractors.
\end{example}

In the previous section we proved that, if the component being eliminated and its targets have no regulator in common,
then to each minimal trap space of the original network corresponds a unique minimal trap space
of the reduced network (\cref{thm:min-ts}).
In particular, under these conditions the attractors of the most permissive dynamics of $f$
and the attractors of the most permissive dynamics of $\rd{f}$ are in one-to-one correspondence.
The same conclusion does \emph{not} hold, in general, for attractors in other dynamics.
An illustration of such scenario is given in \cref{fig:ex:fixed-in-P-1-mediator}.

\section{Phenotype control and reduction}\label{sec:control-reduction}

Recall that, for the purpose of this work, a control strategy is a subspace on which the
dynamics can be restricted to cause some desired effects on the asymptotic dynamics.

Given a Boolean network $f$ and a phenotype $P$, we ask the following questions:
\begin{question}\label{q1}
  If $S$ is a control strategy for $(f, P)$,
  is the subspace $S_{[n-1]}$ a control strategy for $(\rd{f}, P_{[n-1]})$?
  More generally, does $(\rd{f}, P_{[n-1]})$ admit a control strategy?
\end{question}
\begin{question}\label{q2}
  If there exists is a control strategy for $(\rd{f}, P_{[n-1]})$, does $(f, P)$
  admit a control strategy?
\end{question}

We look at answering these questions, in the general case and in the case
of removal of a mediator node.
As explained in \cref{sec:control}, we only consider control strategies where component $n$ is free
(assumption \ref{eq:assumption-Sn-free}) and 
that do not fix any component that is fixed in the phenotype (assumption \ref{eq:assumption-Pi-Si}),
while $n$ can be free or fixed in the phenotype.
The results that we present in this section are summarized in \cref{table:summary}.
We start by discussing the cases that have a positive answer.

\begin{table}
  \begin{subtable}{\linewidth}
  \centering
  \begin{tabular}{l|c|c||c|c|}
    \cline{2-5}
    & \multicolumn{2}{c||}{\makecell{$\exists$ CS for $(f,P)$ $\Rightarrow$\\$\exists$ CS for $(\rd{f},P_{[n-1]})$}} &
      \multicolumn{2}{c|}{\makecell{$\exists$ CS for $(\rd{f},P_{[n-1]})$\\$\Rightarrow$ $\exists$ CS for $(f,P)$}} \\
    \cline{2-5}
    & & $\begin{tikzcd}[column sep=tiny] I \arrow[r] \arrow[rr,bend right=20,gray,"/" marking] & n \arrow[r] & J\end{tikzcd}$ &
      & $\begin{tikzcd}[column sep=tiny] I \arrow[r] \arrow[rr,bend right=20,gray,"/" marking] & n \arrow[r] & J\end{tikzcd}$ \\
    \cline{2-5}
    \hline
    $\AD$ & \multirow{4}*{\xmark Ex.~\ref{ex:fixed-in-P-1}} & \multirow{3}*{\xmark Ex.~\ref{ex:fixed-in-P-1-mediator}} & \multirow{4}*{\xmark Ex.~\ref{ex:fixed-in-P-2}} & \multirow{4}*{\xmark Ex.~\ref{ex:fixed-in-P-2-mediator}} \\
    \cline{1-1}
    $\GAD$ & & & & \\
    \cline{1-1}
    $\SD$ & & & & \\
    \cline{1-1}\cline{3-3}
    MTS & & \checkmark Thm.~\ref{thm:min-ts-control} & & \\
    \cline{1-3}
    VP & \multicolumn{2}{c||}{\checkmark Thm.~\ref{thm:propagation}} & & \\
    \hline
  \end{tabular}
  \caption{$n$ fixed in $P$}\label{table:summary-fixed}
  \bigskip
  \end{subtable}
  \begin{subtable}{\linewidth}
  \centering
  \begin{tabular}{l|c|c||c|c|}
    \cline{2-5}
    & \multicolumn{2}{c||}{\makecell{$\exists$ CS for $(f,P)$ $\Rightarrow$\\$\exists$ CS for $(\rd{f},P_{[n-1]})$}} &
      \multicolumn{2}{c|}{\makecell{$\exists$ CS for $(\rd{f},P_{[n-1]})$\\$\Rightarrow$ $\exists$ CS for $(f,P)$}} \\
    \cline{2-5}
    & & $\begin{tikzcd}[column sep=tiny] I \arrow[r] \arrow[rr,bend right=20,gray,"/" marking] & n \arrow[r] & J\end{tikzcd}$ &
      & $\begin{tikzcd}[column sep=tiny] I \arrow[r] \arrow[rr,bend right=20,gray,"/" marking] & n \arrow[r] & J\end{tikzcd}$ \\
    \cline{2-5}
    \hline
    $\AD$ & \multirow{4}*{\xmark Ex.~\ref{ex:free-in-P-CS-to-no-CS}} & \multirow{3}*{\xmark Ex.~\ref{ex:free-in-P-CS-to-no-CS-mediator}} & \multicolumn{2}{c|}{\xmark Ex.~\ref{ex:new-cs-mediator}} \\
    \cline{1-1}\cline{4-5}
    $\GAD$ & & & \multicolumn{2}{c|}{\xmark Ex.~\ref{ex:new-cs-mediator}} \\
    \cline{1-1}\cline{4-5}
    $\SD$ & & & \multicolumn{2}{c|}{\xmark Ex.\ref{ex:new-cs-mediator}} \\
    \cline{1-1}\cline{3-5}
    MTS & & \checkmark Thm.~\ref{thm:min-ts-control} & \xmark Ex.~\ref{ex:no-CS-to-CS-not-fixed} & \checkmark Thm.~\ref{thm:min-ts-control} \\
    \hline
    VP & \multicolumn{2}{c||}{\checkmark Thm.~\ref{thm:propagation}} & \makecell{\xmark~VP, $\SD$ Ex.~\ref{ex:prop},\ref{ex:free-in-P-no-CS-to-CS}\\ \checkmark~$\AD$, $\GAD$ Thm.~\ref{thm:perc-in-red}} & \checkmark Thm.~\ref{thm:perc-in-red-2} \\
    \hline
  \end{tabular}
  \caption{$n$ free in $P$}\label{table:summary-not-fixed}
  \end{subtable}
  \caption{Summary of results about phenotype control and reduction, for (a) $n$ fixed in the target phenotype $P$
      and (b) $n$ free in the target phenotype $P$.
      We studied whether the existence of a control strategy (CS) in the initial (resp. reduced)
      network always implies the existence of a control strategy in the reduced (resp. initial)
      network. Control strategies target attractors in asynchronous ($\AD$), general asynchronous
      ($\GAD)$, and synchronous ($\SD$) dynamics, as well as minimal trap spaces (MTS).
      VP stands for control by value propagation.
      In each case, we considered any network, and networks where $n$ is a mediator node
      (no regulator of node $n$ regulates a target of $n$).
      The checkmark ($\checkmark$) indicates that the property is always true, whereas
      the cross (\xmark) indicates the existence of counterexamples.}\label{table:summary}
\end{table}

\subsection{Control of minimal trap spaces}\label{sec:control-mts}

\begin{proposition}\label{thm:control-mts}
  Consider an MTS-control strategy $S$ for $(f, P)$ with $S_n = {\star}$.
  Suppose that for each minimal trap space $T'$ of $\rd{C(f, S)}$
  there exists a minimal trap space $T$ of $C(f, S)$ such that 
  $T' \subseteq T_{[n-1]}$.
  Then $S_{[n-1]}$ is an MTS-control strategy for $(\rd{f}, P_{[n-1]})$.
\end{proposition}
\begin{proof}
  By \cref{prop:red-sub}, $C(\rd{f}, S_{[n-1]}) = \rd{C(f, S)}$.
  Since all minimal trap spaces of $C(f,S)$ are contained in $P$,
  we find that all minimal trap spaces of $C(\rd{f},S_{[n-1]})$ are contained in $P_{[n-1]}$.
\end{proof}

The proposition gives a possible strategy to answer \cref{q1} positively.
To answer \cref{q2} positively, we need to ensure that minimal trap spaces cannot ``shrink'' with the reduction,
possibly leading to emergence of some new MTS-control strategies.

\begin{proposition}\label{thm:control-mts-2}
  Consider a Boolean network $f$, a phenotype $P$ with $P_n={\star}$ and an MTS-control strategy $S$ for $(\rd{f}, P_{[n-1]})$.
  Suppose that, for each minimal trap space $T$ for $C(f, S^{\star})$,
  $T_{[n-1]}$ is a minimal trap space for $\rd{C(f, S^{\star})}$.
  Then the subspace $S^{\star}$ is an MTS-control strategy for $(f, P)$.
\end{proposition}
\begin{proof}
  Since $n$ is free in $S^{\star}$, by \cref{prop:red-sub}, $C(\rd{f}, S^{\star}_{[n-1]}=S) = \rd{C(f, S^{\star})}$.
  Given a minimal trap space $T$ for $C(f, S^{\star})$,
  $T_{[n-1]}$ is a minimal trap space contained in $P_{[n-1]}$. Therefore, by definition of $S^{\star}$, $T$ is contained in $P$.
\end{proof}

Observe that, given any subspace $S$, by \cref{rmk:ig-subgraph},
if $n$ is a mediator node for $f$, then $n$ is a mediator node also for $C(f, S)$.
Therefore, combining the results above with \cref{thm:min-ts}, we have the following.

\begin{theorem}\label{thm:min-ts-control}
  Consider a Boolean network $f$ and a phenotype $P$.
  Suppose that no regulator of $n$ regulates a target of $n$.
  \begin{itemize}
    \item[(i)] If $S$ is an MTS-control strategy $S$ for $(f, P)$ with $S_n = {\star}$,
          then $S_{[n-1]}$ is an MTS-control strategy for $(\rd{f}, P_{[n-1]})$.
    \item[(ii)] If $S$ is an MTS-control strategy $S$ for $(\rd{f}, P_{[n-1]})$ and $P_n={\star}$,
          then the subspace $S^{\star}$ is an MTS-control strategy for $(f, P)$.
  \end{itemize}
\end{theorem}

\subsection{Control by value propagation}\label{sec:value-propagation}

Control strategies by value propagation are the strongest.
It is not surprising then that some correspondence can be established between
these strategies and strategies of reduced networks.
We start with a lemma.

\begin{lemma}\label{lemma:propagation}
  If $S$ is a trap space and $\Phi^k_{f}(S)$ is contained in a subspace $P$ for some $k \geq 1$,
  then $\Phi^k_{\rd{f}}(S_{[n-1]})$ is contained in $P_{[n-1]}$.
\end{lemma}
\begin{proof}
  We show, by induction on $k$, that $\Phi^k_{\rd{f}}(S_{[n-1]})$ is contained in
  $(\Phi^k_{f}(S))_{[n-1]}$.

  For all $x \in S_{[n-1]}$, $\sigma(x)$ is in $S$, and $\rd{f}(x)_i = f_i(\sigma(x))$ for all $i \neq n$.
  Therefore, $\Phi_{\rd{f}}(S_{[n-1]})$ is contained in $(\Phi_{f}(S))_{[n-1]}$.

  Now suppose that $\Phi^k_{\rd{f}}(S_{[n-1]})$ is contained in $(\Phi^k_{f}(S))_{[n-1]}$.
  We show that $\Phi^{k+1}_{\rd{f}}(S_{[n-1]})$ is contained in $(\Phi^{k+1}_{f}(S))_{[n-1]}$.
  Since $\sigma(x)$ is in $\Phi^k_{f}(S)$ for all $x$ in $\Phi^k_{\rd{f}}(S_{[n-1]})$,
  we have again that $\Phi^{k+1}_{\rd{f}}(S_{[n-1]}) = \Phi_{\rd{f}}(\Phi^k_{\rd{f}}(S_{[n-1]}))$
  is contained in $(\Phi^{k+1}_f(S))_{[n-1]}$.
\end{proof}

If $S$ is not a trap space, then the lemma might fail, as shown in this simple example.
\begin{example}
Take $f(x_1, x_2) = (x_1 x_2, x_1)$, reducing to $\rd{f}(x_1) = x_1$ by removal of the second component.
Consider $S = {\star} 0$, which is not a trap space. Clearly, $\Phi_f(S)$ is contained in $P = 0{\star}$.
We have $S_{[n-1]} = {\star}$, and $\rd{f}(1) = 1$ which is not contained in $P_{[n-1]} = 0$.
\end{example}

\begin{theorem}\label{thm:propagation}
  Suppose that $S$ is a control strategy by value propagation for $(f,P)$,
  and that $\rd{f}$ is obtained from $f$ by removing a component that is free in $S$.
  Then $S_{[n-1]}$ is a control strategy by value propagation for $(\rd{f}, P_{[n-1]})$.
\end{theorem}
\begin{proof}
  We need to show that $\phi(C(\rd{f}, {S_{[n-1]}}))$ is contained in $P_{[n-1]}$.
  By \cref{prop:red-sub}, we have that $\phi(C(\rd{f}, {S_{[n-1]}})) = \phi(\rd{C(f, S)})$.
  By hypothesis, $\phi(C(f, S)) \subseteq P$, meaning $\Phi^k_{C(f, S)}(\B^n) \subseteq P$ for $k$ sufficiently large.
  By \cref{lemma:propagation}, $\Phi^k_{\rd{C(f, S)}}(\B^{n-1}) \subseteq P_{[n-1]}$ for $k$ sufficiently large,
  which concludes.
\end{proof}

The meaning of the result is that, if we are interested in control by propagation,
a component that is not a candidate target for control can be eliminated,
without loss of control strategies.

\begin{example}\label{ex:prop}
  The existence of a control strategy by value propagation for $(\rd{f}, P_{[n-1]})$
  does not imply the existence of a control strategy by value propagation for $(f, P)$.
  For example, take
  $f(x_1, x_2, x_3) = (x_1 \bar{x}_2 {\vee} x_1 \bar{x}_3 {\vee} x_2 x_3 \bar{x}_1, x_1 \bar{x}_3, x_1)$,
  which reduces to $\rd{f}(x_1, x_2) = (x_1 \bar{x}_2, 0)$.
  For $P=0{\star}{\star}$, there are no control strategies by value propagation
  ($S = {\star} 1 {\star}$ is an MTS-control strategy
  and an attractor-control strategy for asynchronous and general asynchronous dynamics).

  On the other hand, ${\star} 1$ is a control strategy by value propagation for $(\rd{f}, 0{\star})$ (\cref{fig:ex:prop}).

  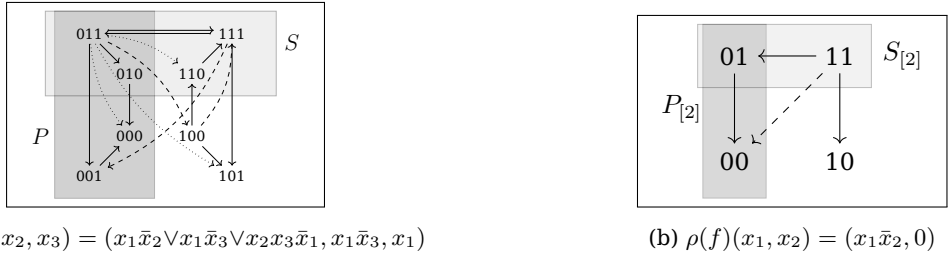
\begin{figure}
  \begin{subfigure}{0.5\textwidth}
  \begin{minipage}{\textwidth}
  \centering
  \fbox{
  \resizebox{3.8cm}{!}{
  \begin{tikzpicture}
  \node at (-0.6-\vd,0*\vd){\Large $P$};
  \filldraw[fill=black!60!white, draw=black, opacity=0.3] (-1.2*\vd,-\vd) rectangle (0.4*\vd,2*\vd);
  \filldraw[fill=black!20!white, draw=black, opacity=0.3] (-1.35*\vd,\td) rectangle (2.35*\vd,2*\vd);
  \node at (2.6*\vd,1.5*\vd){\Large $S$};

  \node (000) at (  0,  0){000};
  \node (100) at (\vd,  0){100};
  \node (010) at (  0,\vd){010};
  \node (110) at (\vd,\vd){110};

  \node (001) at (-\td,-\td){001};
  \node (101) at (\vd+\td,-\td){101};
  \node (011) at ( -\td,\vd+\td){011};
  \node (111) at (\vd+\td,\vd+\td){111};

  \path[->,draw,black]
  (001) edge (000)
  (010) edge (000)
  (011) edge (001)
  (011) edge[transform canvas={yshift=2pt}] (111)
  (011) edge (010)
  (100) edge (101)
  (100) edge (110)
  (110) edge (111)
  (111) edge (101)
  (111) edge[transform canvas={xshift=-2pt}] (011)

  (011) edge[bend left=20,dashed] (100)
  (100) edge[bend right=20,dashed] (111)
  (111) edge[bend left=20,dashed] (001)

  (011) edge[bend right=15,dotted] (101)
  (011) edge[bend left=15,dotted] (110)
  (011) edge[bend right=15,dotted] (000)
  ;
  \end{tikzpicture}
  }
  }
  \end{minipage}%
  \caption{$f(x_1,x_2,x_3) = (x_1 \bar{x}_2 {\vee} x_1 \bar{x}_3 {\vee} x_2 x_3 \bar{x}_1, x_1 \bar{x}_3, x_1)$}
  \end{subfigure}
  \begin{subfigure}{0.5\textwidth}
  \begin{minipage}{\textwidth}
  \centering
  \fbox{
  \begin{tikzpicture}
    \node at (-0.7,0.5*\vd){$P_{[2]}$};
  \filldraw[fill=black!60!white, draw=black, nearly transparent] (-0.3*\vd,-0.35*\vd) rectangle (0.3*\vd,1.3*\vd);
  \filldraw[fill=black!20!white, draw=black, nearly transparent] (-0.35*\vd,+0.7*\vd) rectangle (1.3*\vd,1.3*\vd);
    \node at (1.6*\vd,\vd){$S_{[2]}$};
  \node (00) at (0,0){00};
  \node (10) at (\vd,0){10};
  \node (01) at (0,\vd){01};
  \node (11) at (\vd,\vd){11};
  \path[->,draw,black]
  (01) edge (00)
  (11) edge (01)
  (11) edge (10)
  (11) edge[dashed] (00)
  ;
  \end{tikzpicture}
  }
  \end{minipage}%
  \caption{$\rd{f}(x_1,x_2) = (x_1 \bar{x}_2, 0)$}
  \end{subfigure}
  \caption{Figure for \cref{ex:prop}. $S_{[2]}$ is a control strategy by value propagation, while $S$ is not. $S$ is, however, an attractor-control strategy for asynchronous and general asynchronous dynamics and an MTS-control strategy.}
  \label{fig:ex:prop}
  \end{figure}
\end{example}

We now consider \cref{q2} and show that the existence of a control strategy by value propagation in a reduced network
implies the existence of a control strategy for the original network, which can be computed
from the former (\cref{thm:perc-in-red}).
As shown by the last example, the control strategy for the original network needs not be a
control strategy by value propagation.
The existence of a control strategy by value propagation for the original network is guaranteed however
if the component being eliminated is a mediator (\cref{thm:perc-in-red-2}).

\begin{theorem}\label{thm:perc-in-red}
  If $n$ is free in $P$ and there is a control strategy by value propagation for $(\rd{f}, P_{[n-1]})$, then
  there exists an MTS-control strategy and an attractor-control strategy for $(f,P)$ under dynamics $\AD$ and $\GAD$.
\end{theorem}
\begin{proof}
  By hypothesis, there exists $S \in \subs^{n-1}$ such that $\phi(C(\rd{f}, S)) \subseteq P_{[n-1]}$.
  Call $I$ the set of components fixed in $P$.
  Take any state $y$ in an attractor of $C(\rd{f}, S)$.
  Consider the subspace $Y$ defined as follows:
  $Y_i = {\star}$ for $i \in I \cup \{n\}$, $Y_i = y_i$ otherwise.
  In particular, for all $i$ fixed in $S$, $Y_i$ equals $S_i$ (so $Y_{[n-1]} \subseteq S$),
  and if a component $j \notin I$ is fixed in $\phi(C(\rd{f}, S))$,
  then it is fixed in $Y$ to its propagation value $\phi(C(\rd{f}, S))_j$.
  As a consequence, the subspace $\phi(C(\rd{f}, Y_{[n-1]}))$ is contained in $\phi(C(\rd{f}, S))$,
  and all the attractors of $C(\rd{f}, Y_{[n-1]})$ are contained in $P_{[n-1]}$.

  By \cref{lemma:trap-sets} and \cref{prop:red-sub}, under dynamics $\AD$ and $\GAD$,
  if $A$ is an attractor of $C(f, Y)$,
  there exists at least one attractor $A'$ for
  $\rd{C(f, Y)}=C(\rd{f}, Y_{[n-1]})$ in $A_{[n-1]}$ and,
  for all $u \in A'$, $\sigma(u)$ is in $A$.
  Since all attractors of $C(\rd{f}, Y_{[n-1]})$ are in $P_{[n-1]}$,
  this means that all attractors of $C(f, Y)$
  intersect with $P$.

  Consider any attractor $A$ of $C(f, Y)$ and take a state $z$ in $A \cap P$.
  In particular, $z$ is in $P \cap Y$, and for each $j$ fixed in $\phi(C(\rd{f}, S))$,
  $z_j$ equals $\phi(C(\rd{f}, S))_j$.

  Now take the representative state $w$ of $z$, which is also in $P \cap Y$ (since $P_n=Y_n={\star}$)
  and part of the same attractor.
  Recall that the only free components in $Y$ are $n$ and the components in $I$.
  Since $f_n(w) = w_n$ ($w$ is representative) and
  $f_i(w) = \rd{f}_i(w_{[n-1]}) = P_i = w_i$ for all $i \in I$ ($w_{[n-1]}$ is in $\phi(C(\rd{f}, S)) \subseteq P_{[n-1]}$),
  $w$ is a fixed point.
  Therefore $w$ and $z$ coincide, all the attractors of $C(f, Y)$ are fixed points,
  and all attractors and minimal trap spaces of $C(f, Y)$ are contained in $P$.
\end{proof}

\begin{theorem}\label{thm:perc-in-red-2}
  Consider a subspace $P$ with $P_n = {\star}$.
  If $S$ is a control strategy by value propagation for $(\rd{f}, P_{[n-1]})$
  and no regulator of $n$ regulates a target of $n$,
  then $S^{\star}$ a control strategy by value propagation for $(f,P)$.
\end{theorem}
\begin{proof}
  Set $g = C(f, S^{\star})$. Then by \cref{prop:red-sub} we have $\rd{g} = C(\rd{f}, S)$.
  Define $Y = \phi(C(\rd{f}, S)) = \phi(\rd{g}) \subseteq P_{[n-1]}$ and $Z = \phi(C(f, S^{\star})) = \phi(g)$.

  By \cref{lemma:propagation}, $Y$ is contained in $Z_{[n-1]}$.
  We assume that the subspace $Y$ is strictly smaller than $Z_{[n-1]}$
  and show that there is a regulator of $n$ that regulates a target of $n$.

  By definition of $Z$, we have $\Phi_g(Z) = Z$.
  On the other hand, $\Phi_{\rd{g}}(Z_{[n-1]})$ is strictly contained in $Z_{[n-1]}$,
  meaning that there exists $i$ fixed in $Y$ and not in $Z$
  such that $\rd{g}_i$ is constant on $Z_{[n-1]}$
  and $g_i$ is not constant on $Z$.
  This means in particular that $Z_n$ is not fixed.

  If $g_n$ is constant on $Z$, then $\Phi_g(Z)$ is strictly contained in $Z$, a contradiction.

  Therefore, we can conclude by applying \cref{lemma:regulators} on $g$ and the subspace $Z$
  and invoking \cref{rmk:ig-subgraph}.
\end{proof}

\subsection{Counterexamples, elimination of components fixed in the phenotype}

In the following we show that \cref{q1} and \cref{q2} posed at the start of the section
have negative answers in all the cases not covered by the previous results,
for the removal of a component fixed in the phenotype (\cref{table:summary-fixed}).

\subsubsection{The projection of a control strategy is not a control strategy for the reduction}

\begin{example}\label{ex:fixed-in-P-1}
  Consider first \cref{q1}.
  Take the Boolean network in \cref{fig:ex-all} (a),
  with target phenotype $P = 0{\star} 0$.
  Since the fixed point $000$ is the unique attractor of $f$ (in all state transition graphs),
  the full space $S = {\star}{\star}{\star}$ is an attractor-control strategy under all dynamics
  and an MTS-control strategy for $(f, P)$.

  Now consider the elimination of the third component (\cref{fig:ex-all} (b)).
  The target phenotype becomes $P_{[2]}=0{\star}$.
  The state transition graphs of $\rd{f}$ now admit two minimal trap spaces and two attractors,
  and applying the (trivial) control $S_{[2]}={\star}{\star}$ does not guarantee
  that both minimal trap spaces and both attractors fall in the target phenotype.
  The other possible subspaces ${\star}0$ and ${\star}1$ are also not control strategies.

  Observe that $S$ is not a control strategy by value propagation
  (we saw in \cref{thm:propagation} that
  control strategies by propagation behave well under elimination of components,
  even when the component being eliminated is fixed in the target phenotype).
\end{example}

\begin{example}\label{ex:fixed-in-P-1-mediator}
  We can reconsider \cref{q1} with the additional condition that $n$ is a linear mediator node.

  Take the network $f(x_1, x_2, x_3, x_4) = (x_3 {\vee} x_1 x_2 {\vee} \bar x_1 \bar x_2, x_4 {\vee}
  x_2 \bar x_1, x_3 \bar x_1 {\vee} \bar x_1 \bar x_2, x_3)$
  which reduces to the network $\rd{f}(x_1, x_2, x_3) = (x_3 {\vee} x_1 x_2 {\vee} \bar x_1 \bar x_2, x_3 {\vee} x_2 \bar x_1, x_3 \bar x_1 {\vee} \bar x_1 \bar x_2)$
  (\cref{fig:ex:fixed-in-P-1-mediator}).
  All dynamics of $f$ have only one attractor, the fixed point $0100$,
  while the asynchronous, synchronous and general asynchronous dynamics of $\rd{f}$
  have an additional attractor.

  The attractor of $f$ is contained in the subspace $P=0{\star}{\star}0$.
  However, the reduction $\rd{f}$ does not admit any attractor-control strategy for $P_{[n-1]}=0{\star}{\star}$.

  On the other hand, both $f$ and $\rd{f}$ have only one minimal trap space, contained in the phenotype:
  by \cref{thm:min-ts-control}, in the case of removal of a mediator node,
  the existence of an MTS-control strategy for $f$ guarantees the existence
  of an MTS-control strategy in the reduced network.

  \begin{figure}[t]
  \begin{subfigure}{0.55\textwidth}
  \centering
  \begin{minipage}{0.8\textwidth}
  \centering
  \fbox{\resizebox{\textwidth}{!}{
  \begin{tikzpicture}
  \filldraw[fill=black!60!white, draw=black, opacity=0.3] (-1.5*\vd,-1.2*\vd) rectangle (0.4*\vd,2.2*\vd);
  \node (0000) at (  0,  0){0000};
  \node (1000) at (\vd,  0){1000};
  \node (0010) at (  0,\vd){0010};
  \node (1010) at (\vd,\vd){1010};

  \node (0100) at ( -\vd, -\vd){0100};
  \node (1100) at (2*\vd, -\vd){1100};
  \node (0110) at ( -\vd,2*\vd){0110};
  \node (1110) at (2*\vd,2*\vd){1110};

  \node (0001) at (  0+\sd,  0){0001};
  \node (1001) at (\vd+\sd,  0){1001};
  \node (0011) at (  0+\sd,\vd){0011};
  \node (1011) at (\vd+\sd,\vd){1011};

  \node (0101) at ( -\vd+\sd, -\vd){0101};
  \node (1101) at (2*\vd+\sd, -\vd){1101};
  \node (0111) at ( -\vd+\sd,2*\vd){0111};
  \node (1111) at (2*\vd+\sd,2*\vd){1111};

  \path[->,draw,black]
  (0000) edge (0010)
  (0000) edge[transform canvas={yshift=1.5pt}] (1000)
  (0001) edge (0011)
  (0001) edge (0101)
  (0001) edge[transform canvas={yshift=1.5pt}] (1001)
  (0001) edge[bend left=15] (0000)
  (0010) edge[bend left=15] (0011)
  (0010) edge (1010)
  (0011) edge (0111)
  (0011) edge (1011)
  (0101) edge[bend left=15] (0100)
  (0110) edge (1110)
  (0110) edge[bend left=15] (0111)
  (0111) edge (1111)
  (1000) edge[transform canvas={yshift=-1.5pt}] (0000)
  (1001) edge[transform canvas={yshift=-1.5pt}] (0001)
  (1001) edge (1101)
  (1001) edge[bend left=15] (1000)
  (1010) edge (1000)
  (1010) edge[bend left=15] (1011)
  (1011) edge (1001)
  (1011) edge (1111)
  (1100) edge (1000)
  (1101) edge[bend left=15] (1100)
  (1110) edge (1100)
  (1110) edge[bend left=15] (1111)
  (1110) edge (1010)
  (1111) edge (1101)
  (0000) edge[bend left=15,dashed] (1010)
  (1010) edge[bend left=15,dashed] (1001)
  (0001) edge[bend left=15,dashed] (1110)
  (1110) edge[bend left=15,dashed] (1001)
  (0010) edge[bend left=15,dashed] (1011)
  (1011) edge[bend left=15,dashed] (1101)
  (0011) edge[bend left=15,dashed] (1111)
  (0110) edge[bend left=15,dashed] (1111)
  (1001) edge[bend left=15,dashed] (0100)
  (0001) edge[bend left=15,dotted] (0010)
  (0001) edge[bend left=15,dotted] (0111)
  (0001) edge[bend left=15,dotted] (1011)
  (0001) edge[bend left=15,dotted] (1100)
  (0001) edge[bend left=15,dotted] (0100)
  (0001) edge[bend left=15,dotted] (0110)
  (0001) edge[bend left=15,dotted] (1000)
  (0001) edge[bend left=15,dotted] (1010)
  (0001) edge[bend left=15,dotted] (1101)
  (0001) edge[bend left=15,dotted] (1111)
  (1001) edge[bend left=15,dotted] (0101)
  (1001) edge[bend left=15,dotted] (1100)
  (1001) edge[bend left=15,dotted] (0000)
  (1110) edge[bend left=15,dotted] (1011)
  (1110) edge[bend left=15,dotted] (1101)
  (1110) edge[bend left=15,dotted] (1000)
  ;
  \end{tikzpicture}
  }}
  \end{minipage}%
  \caption{$(x_3 {\vee} x_1 x_2 {\vee} \bar x_1 \bar x_2, x_4 {\vee} x_2 \bar x_1, x_3 \bar x_1 {\vee} \bar x_1 \bar x_2, x_3)$}
  \end{subfigure}\hfill
  \begin{subfigure}{0.4\textwidth}
  \centering
  \begin{minipage}{0.8\textwidth}
  \centering
  \fbox{
  \resizebox{0.65\textwidth}{!}{
  \begin{tikzpicture}
  \filldraw[fill=black!60!white, draw=black, opacity=0.3] (-1.5*\vd,-1.2*\vd) rectangle (0.4*\vd,2.2*\vd);
  \node (000) at (  0,  0){000};
  \node (100) at (\vd,  0){100};
  \node (001) at (  0,\vd){001};
  \node (101) at (\vd,\vd){101};

  \node (010) at ( -\vd, -\vd){010};
  \node (110) at (2*\vd, -\vd){110};
  \node (011) at ( -\vd,2*\vd){011};
  \node (111) at (2*\vd,2*\vd){111};

  \path[->,draw,black]
  (000) edge[transform canvas={yshift=1.5pt}] (100)
  (000) edge (001)
  (001) edge (101)
  (001) edge (011)
  (011) edge (111)
  (100) edge[transform canvas={yshift=-1.5pt}] (000)
  (101) edge (100)
  (101) edge (111)
  (110) edge (100)
  (111) edge (110)
  (000) edge[dashed] (101)
  (101) edge[dashed] (110)
  (001) edge[dashed] (111)
  ;
  \end{tikzpicture}
  }}
  \end{minipage}%
  \caption{$(x_3 {\vee} x_1 x_2 {\vee} \bar x_1 \bar x_2, x_3 {\vee} x_2 \bar x_1, x_3 \bar x_1 {\vee} \bar x_1 \bar x_2)$}
  \end{subfigure}
  \caption{The network on the left has only one attractor, the fixed point 0100.
  The full space is an attractor- and MTS-control strategy for $P=0{\star}{\star}0$.
  The reduced network, on the right, has an additional attractor in the asynchronous, synchronous and general asynchronous dynamics,
  and no attractor-control strategy.}\label{fig:ex:fixed-in-P-1-mediator}
  \end{figure}
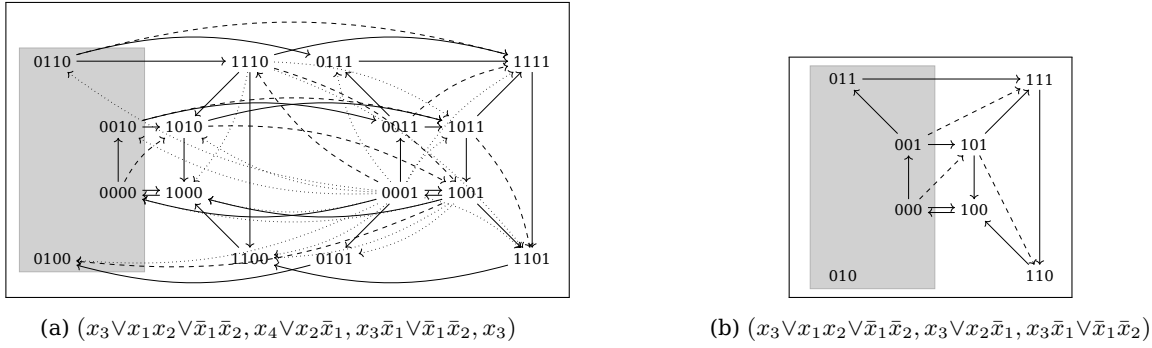
\end{example}

\subsubsection{New control strategies in reduced networks}

\begin{example}\label{ex:fixed-in-P-2}
  Consider \cref{q2}, and again the Boolean network in \cref{fig:ex-all}.
  This time take $P = *01$.
  To find a control strategy, we can consider three possible subspaces:
  ${\star}{\star}{\star}$, $0{\star}{\star}$ and $1{\star}{\star}$.
  The first is clearly not a control strategy, since the unique attractor $000$
  is outside $P$.
  The attractors of the state transition graphs defined by $C(f,0{\star}{\star})$
  and $C(f,1{\star}{\star})$ are also not contained in $P$.
  On the other hand, $S = 0{\star}$ is an attractor-control strategy,
  an MTS-control strategy and a control strategy by propagation for $(\rd{f}, P_{[2]}={\star} 0)$.

  For a minimal example, we could take the simple network $f(x_1, x_2) = (0, 0)$.
  Clearly, no control strategy exists if we consider $P = {\star} 1$.
  On the other hand, $\rd{f} = 0$ and $P_{[1]} = {\star}$,
  so that the full state space $S = {\star}$ is trivially a control strategy
  (it is an attractor-control strategy, an MTS-control strategy and a control strategy by propagation).
\end{example}

\begin{example}\label{ex:fixed-in-P-2-mediator}
  For an example where the component being removed is a linear mediator component, consider
  $f(x_1,x_2,x_3) = (x_2 x_3, 0, x_1)$ and $P=0{\star}1$ (\cref{fig:ex-fixed-P-2-mediator}).
  Without fixing any component, the reduced network is controlled to $P_{[2]} = 0{\star}$;
  on the other hand, there are no control strategies for $P$ in the original network.

  \begin{figure}
  \begin{subfigure}{0.5\textwidth}
  \begin{minipage}{\textwidth}
  \centering
  \fbox{
  \resizebox{3.8cm}{!}{
  \begin{tikzpicture}
  \filldraw[fill=black!60!white, draw=black, opacity=0.3] (-1.2*\vd,-\vd) rectangle (-0.3*\vd,2*\vd);

  \node (000) at (  0,  0){000};
  \node (100) at (\vd,  0){100};
  \node (010) at (  0,\vd){010};
  \node (110) at (\vd,\vd){110};

  \node (001) at (-\td,-\td){001};
  \node (101) at (\vd+\td,-\td){101};
  \node (011) at ( -\td,\vd+\td){011};
  \node (111) at (\vd+\td,\vd+\td){111};

  \path[->,draw,black]
  (001) edge (000)
  (010) edge (000)
  (011) edge (001)
  (011) edge (111)
  (011) edge (010)
  (100) edge (000)
  (100) edge (101)
  (101) edge (001)
  (110) edge (100)
  (110) edge (111)
  (110) edge (010)
  (111) edge (101)

  (011) edge[bend left=15,dashed] (100)
  (100) edge[bend left=15,dashed] (001)
  (110) edge[bend left=15,dashed] (001)

  (011) edge[bend left=25,dotted] (101)
  (011) edge[bend left=15,dotted] (110)
  (011) edge[bend right=15,dotted] (000)
  (110) edge[bend left=15,dotted] (101)
  (110) edge[dotted] (000)
  (110) edge[bend left=15,dotted] (011)
  ;
  \end{tikzpicture}
  }
  }
  \end{minipage}%
  \caption{$f(x_1,x_2,x_3) = (x_2x_3, 0, x_1)$}
  \end{subfigure}
  \begin{subfigure}{0.5\textwidth}
  \begin{minipage}{\textwidth}
  \centering
  \fbox{
  \begin{tikzpicture}
  \filldraw[fill=black!60!white, draw=black, opacity=0.3] (-0.3*\vd,-0.35*\vd) rectangle (0.3*\vd,1.3*\vd);
  \node (00) at (0,0){00};
  \node (10) at (\vd,0){10};
  \node (01) at (0,\vd){01};
  \node (11) at (\vd,\vd){11};
  \path[->,draw,black]
  (01) edge (00)
  (10) edge (00)
  (11) edge (10)
  ;
  \end{tikzpicture}
  }
  \end{minipage}%
  \caption{$\rd{f}(x_1,x_2) = (x_1 x_2, 0)$}
  \end{subfigure}
  \caption{$\B^2$ is a control strategy for $(\rd{f},P_{[2]}=0{\star})$ (under all definitions considered here) and there are no control strategies for $(f, P=0{\star}1)$.}
  \label{fig:ex-fixed-P-2-mediator}
  \end{figure}
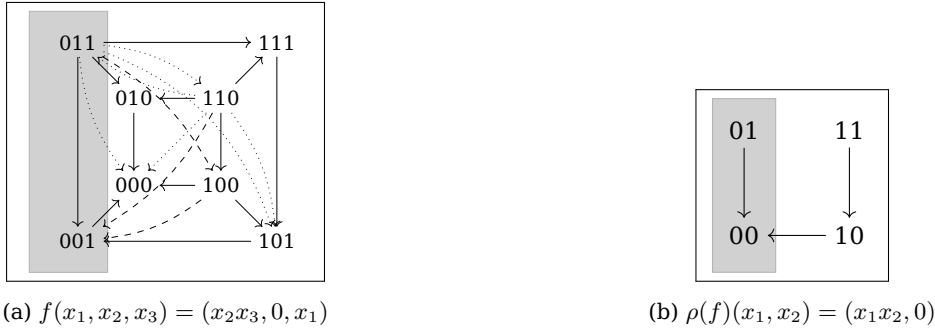
\end{example}

\subsection{Counterexamples, elimination of components not fixed in the phenotype}

The examples in this section cover the negative cases in \cref{table:summary-not-fixed}.

\subsubsection{The projection of a control strategy is not a control strategy for the reduction}

If $S$ is an MTS-control strategy for $(f,P)$,
then $S_{[n-1]}$ is an not necessarily an MTS-control strategy for $(\rd{f}, P_{[n-1]})$.

\begin{example}\label{ex:free-in-P-CS-to-no-CS}
  Take again the Boolean network in \cref{fig:ex-all},
  this time with $P = 0{\star}{\star}$.
  Clearly, $S = {\star}{\star}{\star}$ is an MTS-control strategy and an attractor-control strategy in all dynamics for $(f, P)$,
  but there are no control strategies for $(\rd{f}, P_{[2]} = 0{\star})$.
\end{example}

\begin{example}\label{ex:free-in-P-CS-to-no-CS-mediator}
  For the network in \cref{fig:ex:fixed-in-P-1-mediator} (a), where $n=4$ is a linear mediator component,
  $S = {\star}{\star}{\star}{\star}$ is an attractor-control strategy for $P=0{\star}{\star}{\star}$ in $\SD(f)$, $\AD(f)$ and $\GAD(f)$,
  as well as an MTS-control strategy (0100 is the unique attractor).

  For the network in \cref{fig:ex:fixed-in-P-1-mediator} (b) obtained by eliminating the last component,
  $S = {\star}{\star}{\star}$ is an MTS-control strategy for $P_{[3]}=0{\star}{\star}$ (as guaranteed by \cref{thm:min-ts-control}),
  but not an attractor-control strategy in $\SD(f)$, $\AD(f)$ or $\GAD(f)$.
  One can verify that there are no attractor control strategies for $(\rd{f}, P_{[3]})$.
\end{example}

\subsubsection{New control strategies in reduced networks}

Here we show that, if $S$ is a control strategy for $(\rd{f}, P_{[n-1]})$,
then the subspace $S^{\star}$ is not necessarily a control strategy for $(f, P)$.
The idea is that an attractor or minimal trap space that
in the original network is not fully
contained in $P$ might get reduced to one that is contained in $P_{[n-1]}$.
We first look at an example in dimension 3.

\begin{example}\label{ex:free-in-P-no-CS-to-CS}
Consider the map with dynamics represented in \cref{fig:ex-f} left,
and its reduction after the elimination of the third component, on the right.
Take $P = {\star} 0 {\star}$.
Then $P_{[2]} = {\star} 0$, and $\B^2$ is a control strategy by value propagation
for $(\rd{f}, P_{[2]})$.
However, $\B^3$ is not a control strategy by value propagation for $(f, P)$,
nor an attractor- or MTS-control strategy.
The subspaces ${\star} {\star} 0$ and ${\star} {\star} 1$ also do not define
control strategies for $(f, P)$.

Note that $0{\star}{\star}$ and $1{\star}{\star}$ are attractor-control strategies for $\AD(f)$ and $\GAD(f)$,
as well as MTS-control strategy for $(f, P)$, in line with \cref{thm:perc-in-red},
despite their union not being a control strategy.
They are not control strategies by value propagation or attractor-control strategies for $\SD(f)$.
\end{example}

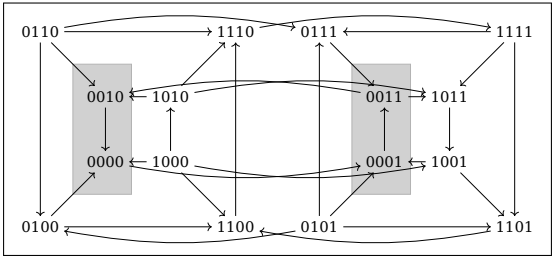
\begin{figure}
\begin{subfigure}{0.59\textwidth}
\begin{minipage}{\textwidth}
\center
\resizebox{4cm}{!}{
\fbox{
  \begin{tikzpicture}
  \filldraw[fill=black!60!white, draw=black, opacity=0.3] (-1.0*\vd,-0.9*\vd) rectangle (2.0*\vd,0.3*\vd);
  \node (000) at (  0,  0){000};
  \node (100) at (\vd,  0){100};
  \node (010) at (  0,\vd){010};
  \node (110) at (\vd,\vd){110};

  \node (001) at (-\td,-\td){001};
  \node (101) at (\vd+\td,-\td){101};
  \node (011) at ( -\td,\vd+\td){011};
  \node (111) at (\vd+\td,\vd+\td){111};

  \path[->,draw,black]
  (000) edge[transform canvas={yshift=2pt}] (100)
  (001) edge[transform canvas={yshift=2pt}] (101)
  (001) edge (000)
  (001) edge[transform canvas={xshift=-2pt}] (011)
  (010) edge (000)
  (010) edge (110)
  (010) edge (011)
  (011) edge[transform canvas={xshift=2pt}] (001)
  (100) edge[transform canvas={yshift=-2pt}] (000)
  (100) edge[transform canvas={xshift=2pt}] (110)
  (100) edge (101)
  (101) edge[transform canvas={yshift=-2pt}] (001)
  (110) edge[transform canvas={xshift=-2pt}] (100)
  (111) edge (101)
  (111) edge (110)
  (111) edge (011)

  (001) edge[dotted,bend right=15] (010)
  (001) edge[dotted,bend right=15] (100)
  (001) edge[dotted,bend right=15] (111)

  (010) edge[dotted,bend right=15] (100)
  (010) edge[dotted,bend right=15] (001)
  (010) edge[dotted,bend right=15] (111)

  (100) edge[dotted,bend right=15] (010)
  (100) edge[dotted,bend right=15] (001)
  (100) edge[dotted,bend right=15] (111)

  (111) edge[dotted,bend right=15] (010)
  (111) edge[dotted,bend right=15] (100)
  (111) edge[dotted,bend right=15] (001)

  (100) edge[dashed,bend right=20] (011)
  (001) edge[dashed,bend right=20] (110)
  (010) edge[dashed,bend right=20] (101)
  (111) edge[dashed,bend right=20] (000)
  ;
  \end{tikzpicture}
}}
\end{minipage}%
\caption{
    $f(x_1,x_2,x_3) =
    (\bar{x}_1 \bar{x}_2 {\vee} \bar{x}_1 \bar{x}_3 {\vee} x_2 \bar{x}_3,
    x_1 \bar{x}_2 \bar{x}_3 {\vee} \bar{x}_1 \bar{x}_2 x_3,
    x_1 \bar{x}_2 {\vee} \bar{x}_1 x_2)$
}
\end{subfigure}
\begin{subfigure}{0.40\textwidth}
\begin{minipage}{\textwidth}
\center
\fbox{
  \begin{tikzpicture}
  \filldraw[fill=black!60!white, draw=black, opacity=0.3] (-0.3*\vd,-0.3*\vd) rectangle (1.3*\vd,0.3*\vd);
  \node (00) at (0,0){00};
  \node (10) at (\vd,0){10};
  \node (01) at (0,\vd){01};
  \node (11) at (\vd,\vd){11};
  \path[->,draw,black]
  (00) edge[transform canvas={yshift=+2pt}] (10)
  (01) edge[transform canvas={yshift=-2pt}] (00)
  (10) edge (00)
  (11) edge (10)
  ;
  \end{tikzpicture}
}
\end{minipage}
\caption{$\rd{f}(x_1,x_2) = (x_1x_2 {\vee} \bar{x}_1\bar{x}_2, 0)$}
\end{subfigure}
\caption{The Boolean network on the left reduces to the one on the right
by elimination of the third component.
The dotted transitions are part of the asynchronous and generalized asynchronous dynamics.
The dashed transitions are found in the synchronous and in the generalized asynchronous dynamics.
All other transitions are common to all dynamics.
${\star}{\star}$ is a control strategy by value propagation for $(\rd{f}, {\star}0)$,
while ${\star}{\star}{\star}$ is not a control strategy for $(f, {\star}0{\star})$.}\label{fig:ex-f}
\end{figure}

We have seen in \cref{thm:perc-in-red} that the existence of a control strategy by value propagation
in the reduced network guarantees the existence of an attractor-control strategy and MTS-control strategy
for the original network.
In the following examples, the reduced network admits an MTS-control strategy which is also an attractor-control strategy
for $\AD(f)$ and $\GAD(f)$; on the other hand, no control strategy exists for $f$.

\begin{example}\label{ex:no-CS-to-CS-not-fixed}
Here we consider a map with 4 components, as in \cref{fig:ex-no}, where again the last component is removed.
For clarity, \cref{fig:ex-no} (a) and (c) only show the asynchronous dynamics,
but the observations also apply to the general asynchronous dynamics.

Take $P=00{\star}{\star}$ as target, which becomes $P_{[3]}=00{\star}$ in the reduction.
$P_{[3]}$ coincides with the unique attractor and the unique minimal trap space of $\rd{f}$,
so $S={\star}{\star}{\star}=\B^3$ is an attractor-control strategy and MTS-control strategy for $(\rd{f},P_{[3]})$.
Observe that $S$ is not a control strategy by value propagation.

On the other hand, it can be verified that no attractor-control and no MTS-control strategy exist for $(f, P)$.
${\star}{\star}{\star}{\star}$ is not a control strategy, because it is the unique minimal trap space,
and the unique attractor has states outside $P$.
The other subspaces to consider, and the attractors they generate, are as in \cref{fig:ex-no} (b).

\begin{figure}[t]
\begin{subfigure}{0.45\textwidth}
\begin{minipage}{0.95\textwidth}
\fbox{\resizebox{\textwidth}{!}{
\begin{tikzpicture}
\filldraw[fill=black!60!white, draw=black, opacity=0.3] (-0.5*\vd,-0.5*\vd) rectangle (0.4*\vd,1.5*\vd);
\filldraw[fill=black!60!white, draw=black, opacity=0.3] (-0.5*\vd+\sd,-0.5*\vd) rectangle (0.4*\vd+\sd,1.5*\vd);
\node (0000) at (  0,  0){0000};
\node (1000) at (\vd,  0){1000};
\node (0010) at (  0,\vd){0010};
\node (1010) at (\vd,\vd){1010};

\node (0100) at ( -\vd, -\vd){0100};
\node (1100) at (2*\vd, -\vd){1100};
\node (0110) at ( -\vd,2*\vd){0110};
\node (1110) at (2*\vd,2*\vd){1110};

\node (0001) at (  0+\sd,  0){0001};
\node (1001) at (\vd+\sd,  0){1001};
\node (0011) at (  0+\sd,\vd){0011};
\node (1011) at (\vd+\sd,\vd){1011};

\node (0101) at ( -\vd+\sd, -\vd){0101};
\node (1101) at (2*\vd+\sd, -\vd){1101};
\node (0111) at ( -\vd+\sd,2*\vd){0111};
\node (1111) at (2*\vd+\sd,2*\vd){1111};

\path[->,draw,black]
(0000) edge[bend right=10] (0001)
(0001) edge (0011)
(0010) edge (0000)
(0011) edge[bend right=10] (0010)
(0011) edge (1011)
(0100) edge (1100)
(0100) edge (0000)
(0101) edge (0001)
(0101) edge[bend left=10] (0100)
(0101) edge (1101)
(0101) edge (0111)
(0110) edge (0100)
(0110) edge (0010)
(0110) edge (1110)
(0110) edge[bend left=10] (0111)
(0111) edge (0011)
(1000) edge (1100)
(1000) edge (1010)
(1000) edge[bend right=10] (1001)
(1000) edge (0000)
(1001) edge (0001)
(1001) edge (1101)
(1010) edge (0010)
(1010) edge (1110)
(1010) edge[bend left=10] (1011)
(1011) edge (1001)
(1100) edge (1110)
(1101) edge[bend left=10] (1100)
(1110) edge[bend left=10] (1111)
(1111) edge (1101)
(1111) edge (0111)
(1111) edge (1011)
;
\end{tikzpicture}
}}
\end{minipage}%
\caption{$(x_2  \bar{x}_3 {\vee} x_2  \bar{x}_4 {\vee} x_3  x_4  \bar{x}_2,
          x_1  \bar{x}_3 {\vee} x_1  \bar{x}_4$,
          $x_1  \bar{x}_4 {\vee} x_4  \bar{x}_1,
          x_2  x_3 {\vee} x_1  \bar{x}_2 {\vee} \bar{x}_2  \bar{x}_3)$}
\end{subfigure}
\begin{subfigure}{0.24\textwidth}
\center
\begin{minipage}{\textwidth}
  \resizebox{0.90\textwidth}{!}{
  \begin{tabular}{|c|c|}
    \hline
    subspace & \makecell{attractors} \\
    \hline
    ${\star}{\star}{\star}{\star}$ & \multicolumn{1}{c|}{\makecell{${\star}{\star}{\star}{\star}\setminus\{0100,0101,$\\$0110,1000,1010\}$}} \\
    \hline
    ${\star}{\star}{\star}0$ & \multicolumn{1}{c|}{$\{000\}, \{111\}$} \\
    \hline
    ${\star}{\star}{\star}1$ & \multicolumn{1}{c|}{$\{110\}$} \\
    \hline
    ${\star}{\star}0{\star}$ & \multicolumn{1}{c|}{$\{001\}, \{110\}$} \\
    \hline
    ${\star}{\star}00$       & \multicolumn{1}{c|}{$\{00\}, \{11\}$} \\
    \hline
    ${\star}{\star}01$       & \multicolumn{1}{c|}{$\{00\}, \{11\}$} \\
    \hline
    ${\star}{\star}1{\star}$ & \multicolumn{1}{c|}{$\{000\}, \{101\}$} \\
    \hline
    ${\star}{\star}10$       & \multicolumn{1}{c|}{$\{00\}, \{11\}$} \\
    \hline
    ${\star}{\star}11$       & \multicolumn{1}{c|}{$\{10\}$} \\
    \hline
  \end{tabular}}
\end{minipage}%
\caption{Subspaces and attractors induced.}
\end{subfigure}
\begin{subfigure}{0.3\textwidth}
\begin{minipage}{\textwidth}
\centering
\fbox{
\resizebox{0.65\textwidth}{!}{
\begin{tikzpicture}
\filldraw[fill=black!60!white, draw=black, opacity=0.3] (-0.5*\vd,-0.5*\vd) rectangle (0.4*\vd,1.5*\vd);
\node (000) at (  0,  0){000};
\node (100) at (\vd,  0){100};
\node (001) at (  0,\vd){001};
\node (101) at (\vd,\vd){101};

\node (010) at ( -\vd, -\vd){010};
\node (110) at (2*\vd, -\vd){110};
\node (011) at ( -\vd,2*\vd){011};
\node (111) at (2*\vd,2*\vd){111};

\path[->,draw,black]
(000) edge[transform canvas={xshift=2pt}] (001)
(001) edge[transform canvas={xshift=-2pt}] (000)
(010) edge (000)
(010) edge (110)
(011) edge (001)
(100) edge (000)
(100) edge (110)
(101) edge (100)
(110) edge[transform canvas={xshift=2pt}] (111)
(111) edge (101)
(111) edge[transform canvas={xshift=-2pt}] (110)
(111) edge (011)
;
\end{tikzpicture}
}}
\end{minipage}%
\caption{$(x_2  \bar{x}_3 {\vee} x_1 x_3 \bar{x}_2, x_1 \bar{x}_3,$\\
           $x_1  x_2  \bar{x}_3 {\vee} x_2  x_3  \bar{x}_1 {\vee} \bar{x}_1  \bar{x}_2  \bar{x}_3)$}
\end{subfigure}
\caption{(a) Asynchronous dynamics of a Boolean network.
(b) Subpaces that can be considered as candidate control strategies for target $00{\star}{\star}$,
and attractors and minimal trap spaces obtained.
(c) Asynchronous dynamics of the Boolean network obtained from the network in (a) by elimination of the fourth component.}\label{fig:ex-no}
\end{figure}
\end{example}

\begin{example}\label{ex:new-cs-mediator}
  Attractor-control strategies can be introduced in the reduction
  under the hypotheses of \cref{thm:min-ts-control}, even when linear mediator components are removed.

  For asynchronous dynamics, take the network
  \[f(x_1,x_2,x_3,x_4) = (x_1 x_2 {\vee} x_1 \bar x_3 {\vee} x_3 \bar x_1 \bar x_2, \bar x_2 \bar x_3 {\vee} x_2 x_3 \bar x_1, \bar x_3 \bar x_4, \bar x_2),\]
  with $P = 0{\star}{\star}{\star}$.
  For general asynchronous, with $P=0{\star}{\star}{\star}{\star}$, a counterexample is given by the network
  \[f(x_1,x_2,x_3,x_4,x_5) = (x_1 x_3{\vee}x_1 \bar x_4{\vee}x_4 \bar x_1 \bar x_3, \bar x_2, x_2{\vee}\bar x_5, x_1 x_2 x_3{\vee}x_1 x_3 \bar x_4{\vee}x_2 x_3 \bar x_4, \bar x_4),\]
  and for synchronous with $P=0{\star}{\star}$, by the network
  \[f(x_1, x_2, x_3) = (x_1 \bar x_2 {\vee} x_2 \bar x_1, \bar x_2 \bar x_3, \bar x_1).\]
\end{example}

\section{Conclusion}

We performed an extensive analysis of the relationship between phenotype control and Boolean network
reduction by component elimination.
We provided examples that clarify that component elimination can disrupt control in most situations.
We also observed that this reduction technique behaves better in relation to control strategies
that work independently of the update scheme.
In particular, we proved that, if the values fixed by the control strategy propagate
through the network until the phenotype subspace is reached,
then the same control strategy works in the reduced network (\cref{thm:propagation}).
Vice versa, if a control strategy by value propagation exists in a reduced network,
under the appropriate conditions (component being removed not fixed in the phenotype)
a control strategy exists for the original network, although it might not necessarily
work by propagating the fixed values (\cref{thm:perc-in-red,thm:perc-in-red-2}).
In addition, we considered the elimination of components under stricter conditions, that is,
when the component being eliminated is not regulated by regulators of its targets.
Under this hypothesis, we demonstrated that minimal trap spaces are preserved by the reduction
(\cref{thm:min-ts}), and thus their control in the original and reduced networks are also related
(\cref{thm:min-ts-control}).
Further work could address the preservation of other properties related to the global
structure of trap spaces.

We limited our analysis to the classical elimination of non-autoregulated components.
Other types of reduction could be considered,
for instance, elimination of negatively regulated components,
which generalizes the elimination of components considered here \citep{schwieger2024reduction}.
The analysis can be extended to other types of control, for example temporal control or
control that acts on interactions \citep{su2020sequential,biane2018causal}.
All models imply a trade-off between complexity and level of detail attained,
while the consequences of simple differences in modelling choices are often difficult to predict.
Given the popularity of the reduction method analysed here,
these types of investigations can serve as useful references
in the context of logical modelling.

\subsection*{Funding}
ET was funded by the Deutsche Forschungsgemeinschaft (DFG, German Research Foundation)
under Germany's Excellence Strategy – The Berlin Mathematics Research Center MATH+
(EXC-2046/1, project ID: 390685689).
LP was funded by the French Agence Nationale pour la Recherche (ANR) in the scope of the
project ``BNeDiction'' (grant number ANR-20-CE45-0001).

\subsection*{Conflict of interest disclosure}

The authors have no conflict of interest to declare.

\subsection*{Code availability}
Several counterexamples referenced in \cref{table:summary} for the relationship of
control strategies between initial and reduced networks have been synthesized automatically
by logic (Answer-Set) programming.
The code is available at \url{https://github.com/pauleve/BN-example-generator}.

\bibliography{biblio}
\bibliographystyle{abbrvnat}

\end{document}